\documentclass[11pt]{article}

\usepackage[margin=1in]{geometry}
\usepackage[T1]{fontenc}
\usepackage[utf8]{inputenc}
\usepackage{lmodern}
\usepackage{microtype}
\usepackage{amsmath,amssymb,amsthm,mathtools}
\usepackage{enumitem}
\usepackage{xcolor}
\usepackage{hyperref}
\usepackage{aliascnt}

\hypersetup{
  colorlinks=true,
  linkcolor=blue!50!black,
  citecolor=blue!50!black,
  urlcolor=blue!50!black
}

\newtheorem{theorem}{Theorem}[section]

\newaliascnt{proposition}{theorem}
\newtheorem{proposition}[proposition]{Proposition}
\aliascntresetthe{proposition}

\newaliascnt{lemma}{theorem}

\aliascntresetthe{lemma}

\newaliascnt{corollary}{theorem}
\newtheorem{corollary}[corollary]{Corollary}
\aliascntresetthe{corollary}

\newaliascnt{definition}{theorem}
\newtheorem{definition}[definition]{Definition}
\aliascntresetthe{definition}

\newaliascnt{example}{theorem}
\newtheorem{example}[example]{Example}
\aliascntresetthe{example}

\newaliascnt{remark}{theorem}
\newtheorem{remark}[remark]{Remark}
\aliascntresetthe{remark}

\newaliascnt{question}{theorem}

\aliascntresetthe{question}

\usepackage[nameinlink,capitalise]{cleveref}

\crefname{theorem}{Theorem}{Theorems}
\Crefname{theorem}{Theorem}{Theorems}
\crefname{proposition}{Proposition}{Propositions}
\Crefname{proposition}{Proposition}{Propositions}
\crefname{lemma}{Lemma}{Lemmas}
\Crefname{lemma}{Lemma}{Lemmas}
\crefname{corollary}{Corollary}{Corollaries}
\Crefname{corollary}{Corollary}{Corollaries}
\crefname{definition}{Definition}{Definitions}
\Crefname{definition}{Definition}{Definitions}
\crefname{example}{Example}{Examples}
\Crefname{example}{Example}{Examples}
\crefname{remark}{Remark}{Remarks}
\Crefname{remark}{Remark}{Remarks}
\crefname{question}{Question}{Questions}
\Crefname{question}{Question}{Questions}

\newcommand{\N}{\mathbb N}

\newcommand{\Filt}{\mathcal F}

\newcommand{\Ext}{\operatorname{Ext}}
\newcommand{\cl}{\operatorname{cl}}
\newcommand{\M}{\mathsf M}
\newcommand{\Runs}{\operatorname{Runs}}
\newcommand{\K}{\mathcal K}
\newcommand{\pow}{\mathcal P}
\newcommand{\must}{\mathsf{must}}
\newcommand{\may}{\mathsf{may}}
\newcommand{\CCS}{\mathsf{CCS}}

\newcommand{\Res}{\mathsf{Res}}
\newcommand{\RunMean}{\mathsf{RunMean}}
\newcommand{\Act}{\mathsf{Act}}

\newcommand{\Inf}{\operatorname{Inf}}
\newcommand{\Leaves}{\operatorname{Leaves}}

\title{A Stone--\v{C}ech Collecting Semantics for Residual Process Behaviour}

\author{%
  Mike Stannett\\
  School of Computer Science\\
  University of Sheffield
}

\date{}

\begin{document}
\maketitle

\begin{abstract}
This paper develops a compact collecting semantics for the residual behaviour left by nonterminating computation.  An infinite execution is observed as a filtered run
\[
  a:T\to X,
\]
where \(T\) carries a filter \(\Filt\) of large sets of times and \(X\) is a Tychonoff space of observations.  Its meaning is the compact set
\[
  \M_X^{\Filt}(a)
  =
  \beta a[\Ext(\Filt)]
  =
  \bigcap_{F\in\Filt}\cl_{\beta X}a[F]
  \subseteq \beta X.
\]
For sequential time this is the tail-cluster set of the stream in the Stone--\v{C}ech compactification of the observation space.  It gives a common semantics to ordinary recurrence, mixed recurrent behaviour, and escape through noncompact parts of the observation space.

The basic theory establishes tail invariance, functoriality under continuous observations, and a temporal reading for clopen observations: containment in the corresponding clopen region of \(\beta X\) is eventual truth, while nonempty intersection is recurrence.  Progress and fairness assumptions are represented by strengthening the time filter.  Relational meanings are obtained by compactifying products, so correlations between observations made along the same asymptotic view of time are retained.

The main application is to residual behaviour in CCS.  Infinite executions are read as streams of residual processes modulo structural congruence.  The resulting semantics distinguishes stable divergence, finite recurrent divergence, mixed recurrence with escape, and escape through unbounded residual growth.  It validates residual-tail laws for prefixing, guarded unfolding, finite choice, and finite prefix-choice forms, while also identifying the boundary of those laws under parallel composition and synchronisation.  Finite observational quotients provide the computational interface to the compact semantics: abstract meanings become recurrent states and strongly connected component calculations, and resource observations detect unbounded escape without requiring individual points of the Stone--\v{C}ech remainder to be inspected.
\end{abstract}

\section{Introduction}

An infinite computation is not exhausted by the sequence of actions it has performed.  After each transition there is also a residual process, namely the process still to be executed.  The long-term shape of the resulting residual sequence can be quite simple, but it need not be.  It may settle to one process, move forever around a finite cycle, return infinitely often to some observation while also doing other things, or leave every finite part of the residual space.  In the last case no residual process need recur, even though a plainly visible quantity, such as the number of active parallel components, grows without bound.

The problem is to give these different residual tails a common semantic form.  In finite quotients the answer should reduce to ordinary recurrence.  For noncompact spaces of residual processes it should still produce a compact value.  And, since such a compact object is too large to be used by inspection, it should behave well under the observations through which it is read.  Stone--\v{C}ech compactification provides exactly this kind of completion: it makes a Tychonoff observation space compact, and its universal property keeps continuous observations compatible with the construction.

A run is written
\[
  a:T\to X,
\]
where \(T\) is a set of time indices and \(X\) is a Tychonoff observation space.  A filter \(\Filt\) on \(T\) specifies which sets of times count as large.  The meaning of the run, relative to \(\Filt\), is
\[
  \M_X^{\Filt}(a)
  =
  \beta a[\Ext(\Filt)]
  \subseteq \beta X,
\]
where \(\Ext(\Filt)\) is the compact set of ultrafilters on \(T\) extending \(\Filt\).  Equivalently,
\[
  \M_X^{\Filt}(a)
  =
  \bigcap_{F\in\Filt}\cl_{\beta X}a[F].
\]
Thus the meaning is the compact set of asymptotic values of the run, computed in the Stone--\v{C}ech compactification of the observation space.

For ordinary streams \(a:\N\to X\), with the cofinite filter on \(\N\), this becomes
\[
  \M_X(a)=\beta a[\N^*]
  =
  \bigcap_{N<\omega}\cl_{\beta X}\{a_n:n\geq N\}.
\]
In a finite discrete space this is just the set of values visited infinitely often.  For a convergent stream it is the ordinary limit.  For the stream \(0,1,2,\ldots\) in the discrete space \(\N\), it is \(\N^*=\beta\N\setminus\N\), recording escape from every finite subset of \(\N\).  Mixed cases are also possible: a stream may return to one value infinitely often while also visiting an unbounded sequence of new values.

The same alternatives occur for residual processes.  In the process examples, residual processes are considered modulo structural congruence.  Thus \([P]\) denotes the structural-congruence class of the process term \(P\); the congruence is the usual CCS one, including the associativity, commutativity and unit laws for parallel composition.  Guarded defining equations are not being added to structural congruence; they are used through the ordinary transition rule for constants.  The full convention is set out in \Cref{sec:ccs}.

A process such as
\[
  A \stackrel{\mathrm{def}}{=} a.A
\]
has a stable infinite execution, whose residual meaning is the singleton \(\{[A]\}\).  A two-state loop
\[
  P\stackrel{\mathrm{def}}{=}a.Q,
  \qquad
  Q\stackrel{\mathrm{def}}{=}b.P
\]
has finite recurrent residual meaning \(\{[P],[Q]\}\).  By contrast,
\[
  G \stackrel{\mathrm{def}}{=} a.(G\mid G)
\]
has an execution whose residuals are
\[
  G,
  \quad
  G\mid G,
  \quad
  G\mid G\mid G,
  \quad\ldots
\]
up to the structural laws for parallel composition.  No residual class need recur in this execution.  A resource observable counting the number of parallel \(G\)-components sends the corresponding compact meaning into \(\N^*\), so unbounded residual growth is visible through an ordinary integer-valued observation.  Nondeterministic examples may combine these behaviours, with different runs contributing stable, recurrent, or escaping meanings.

Semantic information is extracted from the compact meaning by observations.  If \(f:X\to Y\) is continuous, then the meaning of the observed run \(f\circ a\) is the image of \(\M_X^{\Filt}(a)\) under \(\beta f\).  For clopen observations the reading is elementary, but the ambient space matters.  If \(C\subseteq X\) is clopen, the relevant region is its closure \(\cl_{\beta X}C\) in the compactification, equivalently the clopen extension determined by the characteristic map of \(C\).  Inclusion of \(\M_X^{\Filt}(a)\) in this region says that the run is eventually in \(C\).  Nonempty intersection says that the run returns to \(C\) recurrently.  In a discrete observation space, these tests apply to arbitrary predicates.

The general construction is used in five related ways.
\begin{enumerate}[label=(\roman*),leftmargin=2.5em]
\item
\Cref{thm:tail-closure} identifies the meaning with the intersection of tail closures.  This gives compactness and invariance under changes on a small set of times.

\item
\Cref{thm:functoriality} proves that continuous observations commute with meanings.  \Cref{prop:clopen-temporal} and \Cref{thm:discrete-temporal} give the eventual and recurrent readings for clopen, and hence discrete, observations.

\item
Fairness and progress assumptions are handled by replacing the time filter with the filter generated by the original tail filter and the required progress sets.  This restricts the admissible ultrafilters to those respecting the chosen fairness assumptions.

\item
Relational properties are treated by compactifying products.  For simultaneous observations \(a_i:T\to X_i\), the joint meaning lies in \(\beta(X_1\times\cdots\times X_k)\), so correlations along a common asymptotic view of time are not lost.

\item
Finite observational quotients give computable images of the compact meaning.  In a finite abstract transition graph, abstract meanings are exactly the recurrent sets visited by infinite paths; under the edge-fair convention used here they are reachable terminal strongly connected components.
\end{enumerate}

The CCS part then fixes a process universe, a labelled transition relation and a topology on residual processes.  An execution
\[
  P_0\to P_1\to P_2\to\cdots
\]
is read as the stream \(n\mapsto [P_n]\).  This gives a run-sensitive semantics, recording the compact meaning of each infinite execution, and a flattened semantics, recording the union of all such asymptotic points.  The discrete topology distinguishes structural-congruence classes; observation topologies built from Boolean algebras of process predicates give coarser versions.

The CCS instance also proves a small algebra of tail laws.  Prefixing by a single action, visible or silent, does not affect asymptotic residual meaning:
\[
  \RunMean_{\CCS}([\alpha.P])=\RunMean_{\CCS}([P]).
\]
Finite choice is interpreted by union:
\[
  \RunMean_{\CCS}([P+Q])
  =
  \RunMean_{\CCS}([P])\cup\RunMean_{\CCS}([Q]).
\]
Consequently finite prefix-choice wrappers can be collapsed to the union of the meanings of their leaves.  The paper also records the boundary of this normalisation: it is not a congruence for arbitrary CCS contexts.  Parallel composition, restriction and synchronisation can make an otherwise disappearing prefix relevant before it is consumed.

\subsection*{Why Stone--\v{C}ech?}

The Stone--\v{C}ech compactification has three roles in the semantics.  First, it supplies a compact completion in which tails in noncompact observation spaces acquire meanings.  Ordinary limit sets already account for recurrent tail behaviour in compact spaces, and in finite quotients they reduce to the states seen infinitely often.  Residual processes and resource observations need not live in compact spaces, however.  A sequence may then have no recurrent state while still exhibiting a definite asymptotic pattern, such as escape through larger and larger residuals.  Passing from \(X\) to \(\beta X\) gives such a tail a compact value.

Second, \(\beta X\) is universal for continuous observations of a Tychonoff space.  If \(f:X\to Y\) is continuous, the compact meaning of the observed run \(f\circ a\) is obtained by applying \(\beta f\) to the compact meaning of \(a\) in \(\beta X\).  The semantic object in \(\beta X\) therefore determines all continuous observational images at once.  Finite Boolean quotients, modal observation maps and integer-valued resource maps are all read from the same compact tail object by functoriality.

Third, the same universality explains the computational use of the construction.  Individual non-principal ultrafilters need not be inspected.  The information used in applications comes from images of the compact meaning: clopen predicates give eventual and recurrent assertions, finite observational quotients give recurrent abstract states and strongly connected component calculations, and resource observations can witness unbounded escape by landing in a Stone--\v{C}ech remainder such as \(\N^*\).  The compactification is therefore a semantic carrier whose finite and resource-level projections provide the effective interface.

The main contributions are these.
\begin{enumerate}[label=(\roman*),leftmargin=2.5em]
\item
A compact collecting semantics is defined for residual tails.  A filtered run \(a:T\to X\) is assigned the compact set
\[
  \M_X^{\Filt}(a)=\beta a[\Ext(\Filt)]
  =\bigcap_{F\in\Filt}\cl_{\beta X}a[F]
  \subseteq\beta X,
\]
with ordinary sequential streams obtained from the cofinite filter on \(\N\).  This gives a uniform semantics for stable, recurrent, escaping and mixed tail behaviour.

\item
The construction is given an observational reading.  Continuous observations commute with meanings, clopen observations express eventual and recurrent truth, resource maps detect unbounded escape, and the same filtered formulation accommodates nondeterminism, fairness assumptions, stream-transformer descent and relational meanings.

\item
Finite observational quotients provide the computational interface.  A finite Boolean algebra of process observations induces a finite quotient in which the image of a compact meaning is exactly the set of recurrent abstract states.  For finite abstract transition graphs this becomes the usual strongly connected component calculation, with terminal components appearing under the edge-fair convention used here.

\item
The process-theoretic application is a residual-tail semantics for CCS.  Infinite executions are read as streams of residual processes modulo structural congruence.  The resulting semantics validates tail laws for prefixing, guarded unfolding, finite choice and finite prefix-choice forms, distinguishes stable divergence, finite recurrence and unbounded residual growth, and records the boundary of these laws under parallel composition and synchronisation.
\end{enumerate}

The paper is organised as follows.  \Cref{sec:filtered-semantics} defines the filtered Stone--\v{C}ech semantics and proves the tail-closure theorem.  \Cref{sec:observations} proves functoriality and the clopen temporal reading.  \Cref{sec:nondeterminism,sec:fairness,sec:compositionality,sec:relational} treat nondeterminism, fairness, abstract descent and relational meanings.  \Cref{sec:finite-abstractions} develops finite observational quotients and finite graph computation.  \Cref{sec:ccs} gives the CCS residual semantics, the prefix-choice laws, the parallel-composition boundary example, and the unbounded-observable escape criterion.  The last sections discuss mobile calculi, related work, limitations and further directions.

\paragraph{Topological conventions.}
All spaces to which the Stone--\v{C}ech compactification is applied are assumed to be Tychonoff, that is, completely regular and Hausdorff.  Compact spaces are compact Hausdorff spaces.  If \(X\) is Tychonoff, then \(\beta X\) denotes its Stone--\v{C}ech compactification, with \(X\) identified with its canonical dense image in \(\beta X\).  The notation
\[
  X^*=\beta X\setminus X
\]
will be used when no confusion can arise.  When it matters that \(X\) is open in \(\beta X\), this will be stated explicitly; in particular, this openness holds when \(X\) is locally compact.

If a topology on a set of processes is generated by a Boolean algebra of observations, the following convention is used.  Let \(\mathcal B\subseteq\pow(X)\) be a Boolean algebra of subsets of a set \(X\), and give \(X\) the topology having \(\mathcal B\) as a clopen base.  This topology is zero-dimensional, but it is Hausdorff exactly when \(\mathcal B\) separates points.  Thus, in applications, either \(\mathcal B\) is assumed to separate the points under consideration, or \(X\) is replaced by the observational quotient
\[
  X/{\equiv_{\mathcal B}},
\]
where
\[
  x\equiv_{\mathcal B}y
  \quad\Longleftrightarrow\quad
  \forall B\in\mathcal B\ (x\in B\Leftrightarrow y\in B).
\]
The Boolean algebra then separates the quotient points, and the induced topology on \(X/{\equiv_{\mathcal B}}\) is zero-dimensional Hausdorff, hence Tychonoff.

\section{The filtered Stone--\v{C}ech collecting semantics}
\label{sec:filtered-semantics}

We now set up the general construction.  The time parameter is a set equipped with a filter, and observations take values in a Tychonoff space.  Nothing special is built into the notation for sequential time: ordinary streams are recovered simply by taking the cofinite filter on \(\N\).

\begin{definition}[Filtered time]
A filtered time set is a pair \((T,\Filt)\), where \(T\) is a set and \(\Filt\) is a proper filter on \(T\).  Elements of \(\Filt\) are called large sets, or tail sets.
\end{definition}

The main examples are:
\begin{enumerate}[label=(\roman*),leftmargin=2.5em]
\item \(T=\N\), with the cofinite filter;
\item a directed poset \((T,\leq)\), with the cofinal filter generated by the sets \(\{t\in T:t\geq t_0\}\);
\item a set of finite configurations of a concurrent system, ordered by extension, with a cofinal or fairness-refined filter.
\end{enumerate}

The set \(T\) is regarded as discrete.  Hence \(\beta T\) is the Stone space of ultrafilters on \(T\).  For \(A\subseteq T\), write
\[
  \widehat A=\{\mathcal U\in\beta T:A\in\mathcal U\}.
\]
These sets are clopen and form a base for the Stone topology.  Define
\[
  \Ext(\Filt)
  =
  \{\mathcal U\in\beta T:\Filt\subseteq\mathcal U\}
  =
  \bigcap_{F\in\Filt}\widehat F.
\]
Thus \(\Ext(\Filt)\) is the compact set of ultrafilters extending the filter \(\Filt\).

\begin{definition}[Filtered Stone--\v{C}ech meaning]
Let \((T,\Filt)\) be a filtered time set, let \(X\) be a Tychonoff space, and let \(a:T\to X\) be a run.  Since \(T\) is discrete, \(a\) is continuous and extends uniquely to
\[
  \beta a:\beta T\to\beta X.
\]
The Stone--\v{C}ech meaning of \(a\), relative to \(\Filt\), is
\[
  \M_X^{\Filt}(a)
  =
  \beta a[\Ext(\Filt)]
  \subseteq \beta X.
\]
\end{definition}

For \(T=\N\) with the cofinite filter, \(\Ext(\Filt)=\N^*\), the set of free ultrafilters on \(\N\).  In that case we write
\[
  \M_X(a)=\beta a[\N^*].
\]

\begin{theorem}[Tail-closure theorem]
\label{thm:tail-closure}
Let \((T,\Filt)\) be a filtered time set and let \(a:T\to X\) be a run into a Tychonoff space \(X\).  Then
\[
  \M_X^{\Filt}(a)
  =
  \bigcap_{F\in\Filt}\cl_{\beta X} a[F].
\]
In particular, \(\M_X^{\Filt}(a)\) is a nonempty compact subset of \(\beta X\).
\end{theorem}

\begin{proof}
The set \(\Ext(\Filt)\) is a closed subset of the compact space \(\beta T\).  It is nonempty because every proper filter extends to an ultrafilter.  Hence \(\beta a[\Ext(\Filt)]\) is nonempty and compact.

Let \(\mathcal U\in\Ext(\Filt)\).  If \(F\in\Filt\), then \(F\in\mathcal U\).  Since \(\widehat F\) is the closure of \(F\) in \(\beta T\), we have
\[
  \beta a(\mathcal U)\in \beta a[\widehat F]
  \subseteq \cl_{\beta X}a[F].
\]
Thus \(\M_X^{\Filt}(a)\) is contained in the displayed intersection.

Conversely, suppose
\[
  p\in \bigcap_{F\in\Filt}\cl_{\beta X}a[F].
\]
Consider the family of closed subsets of \(\beta T\)
\[
  \{\widehat F:F\in\Filt\}
  \cup
  \{(\beta a)^{-1}(C):C\text{ is a closed neighbourhood of }p\text{ in }\beta X\}.
\]
This family has the finite intersection property.  Indeed, after finitely many filter elements have been intersected, we obtain some \(F_0\in\Filt\).  If \(C\) is a closed neighbourhood of \(p\), then \(p\in\cl_{\beta X}a[F_0]\) implies \(a[F_0]\cap C\neq\varnothing\).  Hence \(\widehat{F_0}\cap(\beta a)^{-1}(C)\neq\varnothing\).  By compactness of \(\beta T\), there is \(\mathcal U\in\beta T\) lying in all members of the family.  Then \(\mathcal U\) extends \(\Filt\), and \(\beta a(\mathcal U)\) lies in every closed neighbourhood of \(p\).  Since \(\beta X\) is Hausdorff, \(\beta a(\mathcal U)=p\).  Thus \(p\in\M_X^{\Filt}(a)\).
\end{proof}

\begin{corollary}[Tail invariance]
\label{cor:tail-invariance}
Let \(a,b:T\to X\) be runs.  Suppose that
\[
  \{t\in T:a(t)=b(t)\}\in\Filt.
\]
Then
\[
  \M_X^{\Filt}(a)=\M_X^{\Filt}(b).
\]
In particular, for streams indexed by \(\N\), changing finitely many terms does not change the meaning.
\end{corollary}

\begin{proof}
Let \(E=\{t\in T:a(t)=b(t)\}\).  Every ultrafilter in \(\Ext(\Filt)\) contains \(E\).  Hence \(\beta a\) and \(\beta b\) agree on \(\Ext(\Filt)\), and their images of this set are equal.
\end{proof}

\begin{example}[Basic sequential tail shapes]
Let \(a:\N\to X\) be a stream.  Then
\[
  \M_X(a)
  =
  \bigcap_{N<\omega}\cl_{\beta X}\{a_n:n\geq N\}.
\]
If \(a_n\to x\) in \(X\), then \(\M_X(a)=\{x\}\).  If \(X\) is a finite discrete space, then \(\M_X(a)\) is the set of values taken infinitely often; in particular, a two-state alternation has the two recurrent states as its meaning.  If \(X=\N\) is discrete and \(a_n=n\), then \(\M_{\N}(a)=\N^*\).  The mixed stream \(a_{2n}=0\), \(a_{2n+1}=n\) has both a recurrent principal point, namely \(0\), and an escaping non-principal part.  These are the stream-level forms of the stable, finite recurrent, growing and mixed residual behaviours used later for processes.
\end{example}

\section{Observations and temporal readings}
\label{sec:observations}

The compact meaning is useful because observations can be applied after the compactification has been taken.  Continuous observations commute with the construction, and clopen observations give the elementary temporal reading used throughout the paper.

\begin{theorem}[Functoriality]
\label{thm:functoriality}
Let \(f:X\to Y\) be a continuous map between Tychonoff spaces.  For every filtered run \(a:T\to X\),
\[
  \M_Y^{\Filt}(f\circ a)
  =
  \beta f[\M_X^{\Filt}(a)].
\]
\end{theorem}

\begin{proof}
Since \(\beta(f\circ a)=\beta f\circ\beta a\),
\[
\begin{aligned}
  \M_Y^{\Filt}(f\circ a)
  &=\beta(f\circ a)[\Ext(\Filt)] \\
  &=(\beta f\circ\beta a)[\Ext(\Filt)] \\
  &=\beta f[\beta a[\Ext(\Filt)]] \\
  &=\beta f[\M_X^{\Filt}(a)].
\end{aligned}
\]
\end{proof}

\begin{example}[Resource observations]
Let \(a:\N\to X\) be a residual stream, and let \(r:X\to\N\) be a continuous resource observation, with \(\N\) discrete.  Functoriality gives
\[
  \M_{\N}(r\circ a)=\beta r[\M_X(a)].
\]
If the observed resource tends to infinity along the run, in the sense that for every \(m\) one eventually has \(r(a_n)>m\), then
\[
  \M_{\N}(r\circ a)\subseteq\N^*.
\]
Thus unbounded growth can be detected by an ordinary resource map even when no individual residual state recurs.  This is the form used later for the growing CCS process.
\end{example}

For a filter \(\Filt\) on \(T\), a subset \(B\subseteq T\) is called \(\Filt\)-positive if
\[
  B\cap F\neq\varnothing
\]
for every \(F\in\Filt\).  Equivalently, \(\Filt\cup\{B\}\) has the finite intersection property.

\begin{proposition}[Clopen temporal reading]
\label{prop:clopen-temporal}
Let \(a:T\to X\) be a filtered run and let \(C\subseteq X\) be clopen.  Then:
\begin{enumerate}[label=(\roman*),leftmargin=2.5em]
\item
\[
  \M_X^{\Filt}(a)\subseteq\cl_{\beta X}C
  \quad\Longleftrightarrow\quad
  a^{-1}(C)\in\Filt.
\]
\item
\[
  \M_X^{\Filt}(a)\cap\cl_{\beta X}C\neq\varnothing
  \quad\Longleftrightarrow\quad
  a^{-1}(C)\text{ is }\Filt\text{-positive}.
\]
\end{enumerate}
\end{proposition}

\begin{proof}
Since \(C\) is clopen, its characteristic map
\[
  \chi_C:X\to\{0,1\}
\]
is continuous, where \(\{0,1\}\) is discrete.  Its extension \(\beta\chi_C:\beta X\to\{0,1\}\) has
\[
  (\beta\chi_C)^{-1}(1)=\cl_{\beta X}C.
\]
For \(\mathcal U\in\Ext(\Filt)\),
\[
  \beta a(\mathcal U)\in\cl_{\beta X}C
  \quad\Longleftrightarrow\quad
  a^{-1}(C)\in\mathcal U.
\]
The rest is the standard ultrafilter test.  A set belongs to every ultrafilter extending \(\Filt\) exactly when it belongs to \(\Filt\); and it belongs to some ultrafilter extending \(\Filt\) exactly when it is \(\Filt\)-positive.
\end{proof}

\begin{remark}
The closure in \Cref{prop:clopen-temporal} is taken in \(\beta X\), not in \(X\).  Thus \(C\) may be closed in \(X\) while \(\cl_{\beta X}C\) contains additional remainder points.  Equivalently, \(\cl_{\beta X}C\) is the clopen subset of \(\beta X\) classified by the extended characteristic map \(\beta\chi_C:\beta X\to\{0,1\}\).
\end{remark}

When \(X\) is discrete, every subset is clopen.  For \(A\subseteq X\), write
\[
  \widehat A=\{\mathcal V\in\beta X:A\in\mathcal V\}.
\]
This is the closure of \(A\) in \(\beta X\).

\begin{corollary}[Eventual and recurrent readings in discrete spaces]
\label{thm:discrete-temporal}
Let \(X\) be discrete, let \(a:T\to X\), and let \(A\subseteq X\).  Then:
\begin{enumerate}[label=(\roman*),leftmargin=2.5em]
\item
\[
  \M_X^{\Filt}(a)\subseteq\widehat A
  \quad\Longleftrightarrow\quad
  a^{-1}(A)\in\Filt.
\]
\item
\[
  \M_X^{\Filt}(a)\cap\widehat A\neq\varnothing
  \quad\Longleftrightarrow\quad
  a^{-1}(A)\text{ is }\Filt\text{-positive}.
\]
\end{enumerate}
For \(T=\N\) with the cofinite filter, these say respectively that \(a_n\in A\) eventually always and that \(a_n\in A\) infinitely often.
\end{corollary}

\begin{definition}[Tail modalities for discrete observations]
For a discrete observation space \(X\) and \(A\subseteq X\), define
\[
  \Box_\infty A
  \quad\text{by}\quad
  \M_X(a)\subseteq\widehat A,
\]
and
\[
  \Diamond_\infty A
  \quad\text{by}\quad
  \M_X(a)\cap\widehat A\neq\varnothing.
\]
For streams indexed by \(\N\), \(\Box_\infty A\) means eventual truth of \(A\), and \(\Diamond_\infty A\) means infinitely-often truth of \(A\).
\end{definition}

\begin{example}[Discrete tail shapes]
In a finite discrete space, the meaning of a stream is exactly the set of states visited infinitely often.  Thus the periodic stream \(a_n=n\bmod k\) into \(\{0,\ldots,k-1\}\) has meaning \(\{0,\ldots,k-1\}\).

In the discrete space \(\N\), the identity stream \(a_n=n\) has \(\M_\N(a)=\N^*\).  It has no recurrent principal value.  The mixed stream
\[
  a_{2n}=0,
  \qquad
  a_{2n+1}=n
\]
has \(0\) as a recurrent principal point and also has an escaping non-principal part contributed by the unbounded odd-indexed values.  These three cases correspond respectively to finite recurrent residual behaviour, pure unbounded residual growth, and a mixture of recurrence and escape.
\end{example}

Sometimes two observations of the same run are related by sharing the same time index rather than by a map between observation spaces.

\begin{definition}[Translation relation]
Let \(a:T\to X\) and \(b:T\to Y\) be two observations of the same filtered time set.  Their Stone--\v{C}ech translation relation is
\[
  R_{a,b}^{\Filt}
  =
  \{(\beta a(\mathcal U),\beta b(\mathcal U)):
      \mathcal U\in\Ext(\Filt)\}
  \subseteq
  \beta X\times\beta Y.
\]
\end{definition}

\begin{proposition}[When the translation relation is functional]
\label{prop:functional-representation}
The relation \(R_{a,b}^{\Filt}\) is the graph of a function
\[
  \tau:\M_X^{\Filt}(a)\to\M_Y^{\Filt}(b)
\]
if and only if, for all \(\mathcal U,\mathcal V\in\Ext(\Filt)\),
\[
  \beta a(\mathcal U)=\beta a(\mathcal V)
  \quad\Longrightarrow\quad
  \beta b(\mathcal U)=\beta b(\mathcal V).
\]
When this holds, \(\tau\) is continuous.
\end{proposition}

\begin{proof}
The displayed condition says exactly that the second coordinate of a point of \(R_{a,b}^{\Filt}\) is determined by the first coordinate.  Thus it is equivalent to \(R_{a,b}^{\Filt}\) being the graph of a function from its first projection to its second projection.  These projections are precisely \(\M_X^{\Filt}(a)\) and \(\M_Y^{\Filt}(b)\).

The relation \(R_{a,b}^{\Filt}\) is compact, since it is the continuous image of the compact space \(\Ext(\Filt)\) under the map
\[
  \mathcal U\mapsto(\beta a(\mathcal U),\beta b(\mathcal U)).
\]
If the relation is functional, then the first projection
\[
  \pi_1:R_{a,b}^{\Filt}\to \M_X^{\Filt}(a)
\]
is a continuous bijection from a compact space to a Hausdorff space, hence a homeomorphism.  The induced function is \(\tau=\pi_2\circ\pi_1^{-1}\), and is continuous.
\end{proof}

\section{Nondeterminism, may/must satisfaction, and refinement}
\label{sec:nondeterminism}

A nondeterministic or concurrent program typically has many runs.  Let \(P\) be a program and let \(\Runs(P)\) denote its set of possible runs, each observed as a map \(a:T\to X\).  We shall use two levels of collection.

\begin{definition}[Run-sensitive and flattened meanings]
The run-sensitive Stone--\v{C}ech semantics of \(P\) is
\[
  \mathcal M_X^{\Filt}(P)
  =
  \{\M_X^{\Filt}(a):a\in\Runs(P)\}
  \subseteq \K(\beta X),
\]
where \(\K(\beta X)\) denotes the hyperspace of nonempty compact subsets of \(\beta X\).  The flattened semantics is
\[
  \bigcup\mathcal M_X^{\Filt}(P)
  =
  \bigcup_{a\in\Runs(P)}\M_X^{\Filt}(a)
  \subseteq \beta X.
\]
\end{definition}

The distinction matters.  Flattening remembers which ultralimit points are possible, but forgets which points came from the same run.  This is the familiar loss of correlation in collecting semantics.

\begin{definition}[May and must satisfaction]
Let \(S\subseteq\beta X\) be a specification region.  For a single run \(a:T\to X\), define
\[
  a\models_{\must}S
  \quad\Longleftrightarrow\quad
  \M_X^{\Filt}(a)\subseteq S,
\]
and
\[
  a\models_{\may}S
  \quad\Longleftrightarrow\quad
  \M_X^{\Filt}(a)\cap S\neq\varnothing.
\]
For a program \(P\), one may then quantify over runs in the usual demonic or angelic ways.
\end{definition}

In the discrete case, with \(S=\widehat A\), \Cref{thm:discrete-temporal} says that must satisfaction is eventual truth, while may satisfaction is infinitely-often truth, for a single run.  For a nondeterministic program, a demonic reading requires all runs to satisfy the relevant must condition; an angelic reading requires some run to satisfy the relevant may condition.  This is close in shape to the may/must distinctions in nondeterministic powerdomain semantics \cite{Plotkin1976,Smyth1978,AbramskyJung1994}.

\begin{definition}[Asymptotic refinement]
For programs \(P,Q\) with flattened meanings in \(\beta X\), define
\[
  P\preceq_{\flat} Q
  \quad\Longleftrightarrow\quad
  \bigcup\mathcal M_X(P)\subseteq\bigcup\mathcal M_X(Q).
\]
For run-sensitive semantics, define
\[
  P\preceq_{\K} Q
  \quad\Longleftrightarrow\quad
  \forall K\in\mathcal M_X(P)\ \exists L\in\mathcal M_X(Q)\ (K\subseteq L).
\]
Other orders are possible, depending on how nondeterminism is read.
\end{definition}

\section{Concurrency and fairness}
\label{sec:fairness}

The filtered formulation is useful when the relevant notion of eventuality is not simply the cofinite filter on \(\N\).  In an interleaving model one may still use \(T=\N\), with a run depending on a chosen schedule.  For partially ordered or true-concurrency models, \(T\) may instead be a set of finite configurations ordered by extension.

For example, in an event-structure model, configurations are finite sets of events which have occurred and are consistent with the causality and conflict relation \cite{WinskelNielsen1995}.  Along a directed family of configurations, the cofinal filter is generated by the sets
\[
  \uparrow t_0=\{t\in T:t\geq t_0\}.
\]
A run \(a:T\to X\) then assigns an observable value to each sufficiently developed configuration.

Fairness and progress assumptions can be represented by strengthening the filter.  Let \((T,\Filt)\) be a filtered time set, and let
\[
  \mathcal E=\{E_i:i\in I\}
\]
be a family of subsets of \(T\).  The reading is that each \(E_i\) is a set of times at which a specified progress, scheduling, or fairness requirement has been met.

\begin{definition}[Generated fair filter]
Suppose that \(\Filt\cup\mathcal E\) has the finite intersection property.  The filter generated by \(\Filt\) and the sets \(E_i\) is
\[
  \Filt_{\mathcal E}
  =
  \{B\subseteq T:
    \text{for some }F\in\Filt
    \text{ and finite }J\subseteq I,
    \ F\cap\bigcap_{j\in J}E_j\subseteq B\}.
\]
The fair meaning of a run \(a:T\to X\) is
\[
  \M_{X,\mathcal E}^{\Filt}(a)
  =
  \M_X^{\Filt_{\mathcal E}}(a).
\]
\end{definition}

\begin{proposition}[Fair ultralimit semantics]
\label{prop:fair-ultralimit}
With notation as above,
\[
  \M_{X,\mathcal E}^{\Filt}(a)
  =
  \{\beta a(\mathcal U):
      \mathcal U\in\Ext(\Filt),
      \ E_i\in\mathcal U\text{ for all }i\in I\}.
\]
This set is nonempty and compact exactly when \(\Filt\cup\mathcal E\) has the finite intersection property.
\end{proposition}

\begin{proof}
By construction, \(\Filt_{\mathcal E}\) is the smallest filter on \(T\) containing \(\Filt\) and all the sets \(E_i\).  Hence an ultrafilter extends \(\Filt_{\mathcal E}\) exactly when it extends \(\Filt\) and contains every \(E_i\).  The displayed formula follows from the definition of filtered meaning.

The filter \(\Filt_{\mathcal E}\) is proper exactly when \(\Filt\cup\mathcal E\) has the finite intersection property.  In that case it has an ultrafilter extension, and the meaning is the continuous image of the compact set \(\Ext(\Filt_{\mathcal E})\).  Conversely, if the finite intersection property fails, no ultrafilter can contain all the required sets.
\end{proof}

The content of a particular fairness condition lies in the choice of the family \(\mathcal E\).  A single set \(E\) may express one progress event, such as occurrence of a particular action beyond a given reference point.  Repeated progress or standard fairness conditions require a family of sets.

\subsection{Unbounded component progress}

Let \(T\) be a set of finite prefixes of an interleaved execution, ordered by extension.  Suppose the system has components indexed by a set \(I\).  For \(i\in I\) and \(k\in\N\), let
\[
  P_{i,k}\subseteq T
\]
be the set of prefixes in which component \(i\) has taken at least \(k\) steps.  Requiring all sets \(P_{i,k}\) to belong to the admissible ultrafilter expresses unbounded progress of component \(i\).  For a family of components \(I_0\subseteq I\), set
\[
  \mathcal E_{\mathrm{prog}}
  =
  \{P_{i,k}:i\in I_0,\ k\in\N\}.
\]
If \(\Filt\cup\mathcal E_{\mathrm{prog}}\) has the finite intersection property, the corresponding progress-fair meaning is
\[
  \M_X^{\Filt_{\mathcal E_{\mathrm{prog}}}}(a).
\]
This condition rules out asymptotic views of the execution in which a component in \(I_0\) is seen to make only bounded progress.

\subsection{Weak and strong fairness families}

A usual weak-fairness, or justice, condition concerns actions or obligations that remain continuously enabled.  A filtered version can be described as follows.  Let \(T\) be a prefix-ordered set of finite executions.  Let \(J\) be a set of obligations.  For \(j\in J\), suppose we have predicates
\[
  \operatorname{En}_j(t)
  \quad\text{and}\quad
  \operatorname{Done}_j(t,u)
\]
where \(t\leq u\) in \(T\).  The readings are that obligation \(j\) is enabled at prefix \(t\), and that \(j\) is discharged somewhere between \(t\) and \(u\).

For \(j\in J\) and \(t_0\in T\), define
\[
  W_{j,t_0}
  =
  \{u\in T:u\geq t_0
      \text{ and either }
      \exists v\ (t_0\leq v\leq u\text{ and }\neg\operatorname{En}_j(v))
      \text{ or }
      \operatorname{Done}_j(t_0,u)\}.
\]
Thus \(W_{j,t_0}\) contains those later prefixes \(u\) at which either the continuous enabledness of \(j\) from \(t_0\) to \(u\) has failed, or the obligation has been discharged by \(u\).  The weak-fairness family is
\[
  \mathcal E_{\mathrm{wf}}
  =
  \{W_{j,t_0}:j\in J,
  \, t_0\in T\}.
\]
When \(\Filt\cup\mathcal E_{\mathrm{wf}}\) has the finite intersection property, the corresponding weak-fair meaning of \(a:T\to X\) is \(\M_X^{\Filt_{\mathcal E_{\mathrm{wf}}}}(a)\).

Strong fairness, or compassion, can be expressed by changing the indexing family.  For example, let \(En_{j,k}\subseteq T\) be the set of prefixes by which obligation \(j\) has been enabled at least \(k\) times, and let \(Do_{j,k}\subseteq T\) be the set of prefixes by which \(j\) has been discharged at least \(k\) times.  A strong-fairness family may include requirements
\[
  S_{j,k}
  =
  (T\setminus En_{j,k})\cup Do_{j,k}.
\]
Requiring all \(S_{j,k}\) says that an asymptotic view which sees arbitrarily many enablings of \(j\) must also see arbitrarily many discharges of \(j\), subject to consistency of the resulting filter.

\subsection{Fair observations}

The clopen temporal reading applies unchanged after the filter has been strengthened.  If \(C\subseteq X\) is clopen, then
\[
  \M_X^{\Filt_{\mathcal E}}(a)\subseteq\cl_{\beta X}C
\]
holds exactly when \(a^{-1}(C)\in\Filt_{\mathcal E}\), and
\[
  \M_X^{\Filt_{\mathcal E}}(a)\cap\cl_{\beta X}C\neq\varnothing
\]
holds exactly when \(a^{-1}(C)\) is positive with respect to \(\Filt_{\mathcal E}\).

\section{Compositionality}
\label{sec:compositionality}

Only some stream operations descend to Stone--\v{C}ech meanings.  Pointwise continuous postprocessing does, by \Cref{thm:functoriality}.  For other operations the following elementary criterion is useful.

Let \(S_X\) be a class of \(X\)-valued streams and \(S_Y\) a class of \(Y\)-valued streams.  A stream transformer is a map \(F:S_X\to S_Y\).  Let
\[
  q_X:S_X\to\K(\beta X),
  \qquad q_X(a)=\M_X(a),
\]
and similarly for \(q_Y\).

\begin{theorem}[Descent criterion]
\label{thm:descent}
A stream transformer \(F:S_X\to S_Y\) induces a well-defined map
\[
  \overline F:q_X[S_X]\to q_Y[S_Y]
\]
satisfying
\[
  \overline F(\M_X(a))=\M_Y(F(a))
\]
for all \(a\in S_X\) if and only if
\[
  \M_X(a)=\M_X(b)
  \quad\Longrightarrow\quad
  \M_Y(F(a))=\M_Y(F(b))
\]
for all \(a,b\in S_X\).
\end{theorem}

\begin{proof}
If such a map \(\overline F\) exists, the implication is immediate.  Conversely, if the implication holds, define
\[
  \overline F(K)=\M_Y(F(a))
\]
where \(a\) is any stream such that \(K=\M_X(a)\).  The implication says exactly that this definition is independent of the choice of \(a\).
\end{proof}

\begin{example}[Time-sensitive operations]
Let \(X=\{0,1\}\), and let \(F\) output the even-indexed subsequence.  Two streams may have the same meaning \(\{0,1\}\) while their even subsequences have different meanings.  Thus \(F\) need not descend to Stone--\v{C}ech meanings.
\end{example}

\section{Relational meanings}
\label{sec:relational}

Many program properties are relational.  Determinism, noninterference, refinement, confluence, agreement between replicas, and absence of certain races are not simply unary properties of individual states.  They can be treated here by taking meanings of tuples of runs.

Given streams \(a_i:T\to X_i\) for \(i=1,\ldots,k\), define
\[
  \langle a_1,\ldots,a_k\rangle:T\to X_1\times\cdots\times X_k
\]
by pointwise tupling.  The relational meaning is
\[
  \M_{X_1\times\cdots\times X_k}^{\Filt}(a_1,\ldots,a_k)
  =
  \M_{X_1\times\cdots\times X_k}^{\Filt}(\langle a_1,\ldots,a_k\rangle)
  \subseteq
  \beta(X_1\times\cdots\times X_k).
\]

The product is formed before compactification, so correlations along the same ultrafilter on time are retained.

\begin{example}[Correlation lost by products]
Let
\[
  a_n=n\bmod 2,
  \qquad
  b_n=1-a_n,
\]
with values in \(\{0,1\}\).  Individually, \(\M(a)=\M(b)=\{0,1\}\).  But the paired stream has meaning
\[
  \M_{\{0,1\}^2}(a,b)=\{(0,1),(1,0)\}.
\]
The product of the individual meanings is the four-point set \(\{0,1\}^2\), which includes \((0,0)\) and \((1,1)\).  Those two pairs are not present along the paired run.
\end{example}

For a relation \(R\subseteq X^k\) in a discrete observation space, \Cref{thm:discrete-temporal} says that
\[
  \M_{X^k}(a_1,\ldots,a_k)\subseteq\widehat R
\]
means eventual satisfaction of \(R(a_1(n),\ldots,a_k(n))\), while nonempty intersection with \(\widehat R\) means recurrent satisfaction.

\section{Finite observational abstractions}
\label{sec:finite-abstractions}

This section gives the computational interface to the compact semantics.  The Stone--\v{C}ech meaning is a compact semantic carrier for tail behaviour, but its effective use is through continuous observations.  After a finite family of clopen predicates has been chosen, every run has a finite quotient trace, and the meaning of that quotient trace is exactly the image of the original compact meaning under the induced finite observation map.  The finite calculation is therefore a quotient of the compact semantics, not a separate approximation scheme.

In practice one does not try to enumerate points of \(\beta X\), or points of the remainder \(X^*\).  Instead one asks finite or resource-level questions about the run: which process predicates recur, which abstract residual states are eventually unavoidable, whether a resource observation escapes every finite bound, and how fairness changes the recurrent part.  The functoriality theorem says that all these answers are obtained by applying continuous observations to the same compact tail object.

Let \(X\) be a Tychonoff space, and let
\[
  \mathcal A=\{A_1,\ldots,A_m\}
\]
be a finite family of clopen observable predicates on \(X\).  Let \(\mathcal B_{\mathcal A}\) be the finite Boolean algebra generated by the sets \(A_i\).  Define an equivalence relation on \(X\) by
\[
  x\equiv_{\mathcal A}y
  \quad\Longleftrightarrow\quad
  \forall B\in\mathcal B_{\mathcal A}\ (x\in B\Leftrightarrow y\in B).
\]
Let
\[
  q_{\mathcal A}:X\to X_{\mathcal A}=X/{\equiv_{\mathcal A}}
\]
be the quotient map.  The set \(X_{\mathcal A}\) is finite and is given the discrete topology.  Since the predicates in \(\mathcal B_{\mathcal A}\) are clopen, the map \(q_{\mathcal A}\) is continuous.

\begin{proposition}[Finite observational abstraction]
\label{prop:finite-observational-abstraction}
Let \(a:T\to X\) be a filtered run, and let \(q_{\mathcal A}:X\to X_{\mathcal A}\) be the finite observational quotient defined above.  Then
\[
  \M_{X_{\mathcal A}}^{\Filt}(q_{\mathcal A}\circ a)
  =
  \beta q_{\mathcal A}[\M_X^{\Filt}(a)].
\]
Moreover, since \(X_{\mathcal A}\) is finite discrete,
\[
  \M_{X_{\mathcal A}}^{\Filt}(q_{\mathcal A}\circ a)
  =
  \{C\in X_{\mathcal A}:(q_{\mathcal A}\circ a)^{-1}(C)
      \text{ is }\Filt\text{-positive}\}.
\]
For sequential time, this is the set of abstract states visited infinitely often.
\end{proposition}

\begin{proof}
The first statement is \Cref{thm:functoriality}, applied to the continuous map \(q_{\mathcal A}:X\to X_{\mathcal A}\).  Since \(X_{\mathcal A}\) is finite discrete, \(\beta X_{\mathcal A}=X_{\mathcal A}\).  A singleton \(\{C\}\subseteq X_{\mathcal A}\) belongs to the meaning exactly when its preimage under \(q_{\mathcal A}\circ a\) is \(\Filt\)-positive, by \Cref{thm:discrete-temporal}.  This proves the second statement.
\end{proof}

Thus a finite observational abstraction computes the recurrent abstract states of a run.  The abstraction loses distinctions inside each atom of the Boolean algebra, but it preserves exactly the tail information expressible by the chosen observations.  In particular, any verification or analysis carried out in the finite quotient can be lifted back to a statement about the original Stone--\v{C}ech meaning by applying \(\beta q_{\mathcal A}\).  This is the sense in which finite-state calculations below are computations of observable images of the compact semantics.

\subsection{Graph computation for finite transition systems}

The preceding proposition becomes algorithmic when the finite quotient is presented by a finite abstract transition graph.  Let
\[
  G=(A,E)
\]
be a finite directed graph, where \(A\) is a finite set of abstract states and \(E\subseteq A\times A\) is the abstract transition relation.  An infinite abstract run is an infinite path
\[
  \pi=a_0a_1a_2\cdots
\]
with \((a_n,a_{n+1})\in E\) for every \(n\).  Write
\[
  \Inf(\pi)=\{a\in A:a_n=a\text{ for infinitely many }n\}.
\]

\begin{proposition}[Finite graph semantics]
\label{prop:finite-graph-semantics}
Let \(G=(A,E)\) be a finite directed graph and let \(\pi:\N\to A\) be an infinite path.  Then
\[
  \M_A(\pi)=\Inf(\pi).
\]
The set \(\Inf(\pi)\) is nonempty, and after deleting a finite prefix the path remains inside \(\Inf(\pi)\).  Moreover, \(\Inf(\pi)\) is strongly connected as a subgraph of \(G\).  Consequently every abstract Stone--\v{C}ech meaning of an infinite path is a nonempty strongly connected set of abstract states contained in a reachable strongly connected component of \(G\).
\end{proposition}

\begin{proof}
The equality \(\M_A(\pi)=\Inf(\pi)\) is \Cref{prop:finite-observational-abstraction} in the special case where the observation space is already finite and discrete.  Since \(A\) is finite and \(\pi\) is infinite, at least one state occurs infinitely often, so \(\Inf(\pi)\neq\varnothing\).  Every state outside \(\Inf(\pi)\) occurs only finitely often.  Since there are only finitely many such states, there is \(N\) such that \(a_n\in\Inf(\pi)\) for all \(n\geq N\).

Let \(u,v\in\Inf(\pi)\).  Choose an occurrence of \(u\) after \(N\), and then a later occurrence of \(v\).  The segment of the path between these two occurrences lies entirely in \(\Inf(\pi)\), and gives a path from \(u\) to \(v\) inside the induced subgraph on \(\Inf(\pi)\).  Reversing the roles of \(u\) and \(v\) gives a path from \(v\) to \(u\).  Hence \(\Inf(\pi)\) is strongly connected.
\end{proof}

Conversely, a reachable strongly connected part of a finite graph can be realised as the meaning of an infinite path whenever it supports an infinite path visiting all of its vertices.  In particular, if \(C\) is a reachable strongly connected component with more than one vertex, or a one-vertex component with a self-loop, then there is an infinite path whose abstract meaning is exactly \(C\).  One reaches \(C\), and then repeats a closed walk through all vertices of \(C\).  If deadlocked vertices are to count as stable terminal behaviour, one may add stuttering self-loops before applying this statement.

This gives a concrete finite procedure.  Once the finite graph has been built, one computes the reachable strongly connected regions that can support infinite paths.  For a particular path, the meaning is the set of vertices occurring infinitely often.  For a family of paths, the possible finite images of the compact meanings are the recurrent strongly connected sets selected by those paths.  The graph calculation therefore computes the finite observational shadows of the Stone--\v{C}ech meanings.

\subsection{Fair finite graph computation}

The graph computation also gives a finite version of the filter treatment of fairness.  The following result uses a particular edge-fair reading, fixed here to avoid ambiguity.  An infinite path \(\pi=a_0a_1a_2\cdots\) in a finite graph \(G=(A,E)\) is called edge-fair if, whenever a state \(u\) occurs infinitely often along \(\pi\), every edge \((u,v)\in E\) is traversed infinitely often along \(\pi\).

\begin{proposition}[Edge-fair finite graph meanings]
\label{prop:edge-fair-finite-graph}
Let \(G=(A,E)\) be a finite directed graph and let \(\pi\) be an edge-fair infinite path.  Then \(\M_A(\pi)=\Inf(\pi)\), and \(\Inf(\pi)\) is a reachable terminal strongly connected component of \(G\): it is strongly connected, reachable from the initial state of \(\pi\), and has no outgoing edge to \(A\setminus\Inf(\pi)\).

Conversely, if \(C\) is a reachable terminal strongly connected component which supports an infinite path, then there is an edge-fair infinite path which eventually remains in \(C\) and whose meaning is exactly \(C\).  In particular this converse applies to every reachable terminal component after adding stuttering self-loops at deadlocked terminal states.
\end{proposition}

\begin{proof}
The equality \(\M_A(\pi)=\Inf(\pi)\) and strong connectedness are given by \Cref{prop:finite-graph-semantics}.  Suppose \(u\in\Inf(\pi)\) and \((u,v)\in E\).  Since \(\pi\) is edge-fair, the edge \((u,v)\) is traversed infinitely often.  Hence \(v\in\Inf(\pi)\).  Thus there is no edge from \(\Inf(\pi)\) to its complement.  Since \(\Inf(\pi)\) is strongly connected and reachable along \(\pi\), it is a reachable terminal strongly connected component.

For the converse, reach \(C\) along a finite path.  Since \(C\) is strongly connected and supports an infinite path, one can choose an infinite walk inside \(C\) which traverses every edge of the finite subgraph induced by \(C\) infinitely often; for example, enumerate the finitely many edges of \(C\), and between successive edges use strong connectedness to connect the current endpoint to the next required source.  Repeating this enumeration gives an edge-fair path.  Since \(C\) is terminal, all outgoing edges from states in \(C\) remain in \(C\).  Every vertex of \(C\) occurs infinitely often, so the meaning is \(C\).
\end{proof}

Under this edge-fairness convention, SCC computation has a particularly simple reading.  The fair finite meanings are the reachable terminal components selected by fair infinite paths, subject only to the convention about stuttering at deadlocked terminal states.  Thus the same finite quotient can be analysed either under an unrestricted scheduler, where recurrent strongly connected sets may occur, or under the edge-fair convention above, where terminal strongly connected components are forced.  In both cases the graph-theoretic output is an observable image of the same underlying compact tail semantics.

\subsection{A small CCS-shaped abstraction}

For a process space \(\mathcal P\), choose finitely many process observations.  In the discrete process topology these predicates are automatically clopen.  In a coarser observation topology they are clopen precisely when they belong to, or are generated by, the chosen Boolean algebra of observations.  For example:
\[
\begin{array}{ll}
A_1 &= \{P:P\text{ has an }a\text{-transition}\},\\[2mm]
A_2 &= \{P:P\text{ has a }b\text{-transition}\},\\[2mm]
A_3 &= \{P:P\text{ has a }\tau\text{-transition}\},\\[2mm]
A_4 &= \{P:P\text{ has no outgoing transition}\}.
\end{array}
\]
These generate a finite Boolean algebra of observations and hence a finite quotient
\[
  q_{\mathcal A}:\mathcal P\to\mathcal P_{\mathcal A}.
\]
An infinite residual execution
\[
  P_0\to P_1\to P_2\to\cdots
\]
then gives an abstract residual execution in the finite set \(\mathcal P_{\mathcal A}\).  By \Cref{prop:finite-observational-abstraction}, the abstract tail meaning is the image of the concrete residual meaning under the induced map \(\beta q_{\mathcal A}\).

The abstract meaning is therefore the set of observation-types occurring infinitely often.  If the abstract meaning is contained in the class where \(A_1\) holds, then after some finite point every residual process has an \(a\)-transition.  If the abstract meaning intersects the class where \(A_4\) holds, then deadlocked residuals occur infinitely often.  If the abstract meaning contains several atoms, the run has recurrently different observable process types, even if the concrete residual processes themselves do not recur.

If a finite abstract transition graph has been computed for these observation types, \Cref{prop:finite-graph-semantics} and \Cref{prop:edge-fair-finite-graph} reduce the possible abstract tail meanings to graph analysis.  Without a fairness assumption one computes the recurrent strongly connected sets that can be visited by infinite abstract paths.  With edge fairness in the sense just defined, the possible fair meanings are the reachable terminal strongly connected components, subject to the convention about stuttering at deadlocked terminal states.  In either case the finite result is a computable projection of the compact residual meaning rather than a separate semantics.

\section{Residual behaviour in CCS}
\label{sec:ccs}

We now turn to the main process-calculus application of the construction.  The observation made along an execution is not just the next action, but the residual process left after each transition.  Thus an infinite execution
\[
  P_0\to P_1\to P_2\to\cdots
\]
is read as a stream of structural-congruence classes of CCS processes, and its Stone--\v{C}ech meaning records the residual processes, or residual process types, visible along the infinite tail.

The section has three parts.  First, it fixes the process universe, the labelled transition relation, and the process topologies used for residual observation.  The primary topology is the discrete topology on structural-congruence classes.  Coarser observation-induced topologies are obtained from Boolean algebras of modal process observations, passing to an observational quotient when separation is needed.  Secondly, the general tail theory gives concrete process laws: finite prefixing, guarded unfolding, finite choice and finite prefix-choice forms have simple asymptotic residual meanings.  Thirdly, the section records the boundary of these laws.  They are laws of residual tails, not congruence laws for arbitrary process contexts; in particular, parallel composition and synchronisation can make a finite prefix relevant before it disappears from the tail.

The CCS presentation itself is standard.  Processes are closed guarded CCS terms, considered modulo structural congruence, with the usual strong labelled transition relation.  Standard references for CCS and its labelled transition semantics include Milner's book \cite{Milner1989}.

\subsection{The process universe}

Let \(\Act\) be a set of visible action names, with co-names written \(\overline a\) for \(a\in\Act\), and let \(\tau\) be the silent action.  Labels are drawn from
\[
  \Lambda=\Act\cup\overline{\Act}\cup\{\tau\}.
\]
We use the following CCS grammar:
\[
  P,Q ::= 0
  \mid \alpha.P
  \mid P+Q
  \mid P\mid Q
  \mid P\setminus L
  \mid P[f]
  \mid A,
\]
where \(\alpha\in\Lambda\), \(L\subseteq\Act\) is finite, \(f\) is a relabelling function respecting co-names and \(\tau\), and \(A\) ranges over process constants.  Each constant \(A\) is equipped with a guarded defining equation
\[
  A \stackrel{\mathrm{def}}{=} P_A.
\]
Only closed terms are considered.

Let \(\equiv\) be the usual structural congruence generated by the monoid laws for parallel composition,
\[
  P\mid 0\equiv P,
  \qquad
  P\mid Q\equiv Q\mid P,
  \qquad
  (P\mid Q)\mid R\equiv P\mid(Q\mid R),
\]
together with the corresponding congruence rules for the syntax under consideration.  Write \([P]\) for the structural-congruence class of \(P\), and let \(\CCS\) denote the set of closed guarded CCS processes modulo \(\equiv\).  Defining equations for guarded constants are not included in structural congruence; they are used only through the standard transition rule for constants.

The labelled transition relation is the standard strong CCS transition relation, written
\[
  [P]\xrightarrow{\mu}[Q],
  \qquad \mu\in\Lambda.
\]
The rules include prefix, choice, parallel interleaving, synchronisation giving a \(\tau\)-transition, restriction, relabelling, and unfolding of guarded constants.  Since the relation is invariant under structural congruence, it is well defined on equivalence classes.

\begin{remark}
The examples below use only prefix, guarded recursion, choice and parallel composition.  Restriction, relabelling and synchronisation are included so that the process universe is the usual CCS one rather than a special fragment.
\end{remark}

\subsection{Residual runs and residual meanings}

\begin{definition}[Residual run]
Let \([P]\in\CCS\).  An infinite residual run from \([P]\) is a sequence
\[
  \rho=([P_0],[P_1],[P_2],\ldots)
\]
such that \([P_0]=[P]\) and, for every \(n\), there is a label \(\mu_n\in\Lambda\) with
\[
  [P_n]\xrightarrow{\mu_n}[P_{n+1}].
\]
We write \(\Runs_\omega([P])\) for the set of infinite residual runs from \([P]\).
\end{definition}

The run \(\rho\) is viewed as a stream
\[
  \rho:\N\to\CCS,
  \qquad
  n\mapsto [P_n].
\]
With the discrete topology on \(\CCS\), its Stone--\v{C}ech residual meaning is
\[
  \M_{\CCS}(\rho)
  =
  \beta\rho[\N^*]
  =
  \bigcap_{N<\omega}\cl_{\beta\CCS}\{[P_n]:n\geq N\}.
\]

\begin{definition}[Process-level residual semantics]
For a process \([P]\), define its run-sensitive residual semantics by
\[
  \RunMean_{\CCS}([P])
  =
  \{\M_{\CCS}(\rho):\rho\in\Runs_\omega([P])\}
  \subseteq
  \K(\beta\CCS).
\]
The corresponding flattened residual set is
\[
  \Res_\beta([P])
  =
  \bigcup_{\rho\in\Runs_\omega([P])}\M_{\CCS}(\rho)
  \subseteq
  \beta\CCS.
\]
If \([P]\) has no infinite residual run, then \(\RunMean_{\CCS}([P])=\varnothing\).
\end{definition}

Finite terminating computations are therefore absent from \(\RunMean_{\CCS}([P])\).  If termination is to be represented as eventual stability, one can add stuttering self-loops at terminal processes before applying the construction.

\subsection{Asymptotic equivalence and prefix erasure}
\label{subsec:ccs-prefix-erasure}

The residual semantics is insensitive to finite initial behaviour.  The laws in this subsection should therefore be read as laws of infinite residual tails, not as equations in the ordinary labelled transition system.  Their proofs use the Stone--\v{C}ech construction only through tail invariance: once two residual streams agree after finitely many steps, they have the same asymptotic meaning.

\begin{definition}[Asymptotic residual equivalence]
For closed guarded CCS processes \([P]\) and \([Q]\), define
\[
  [P]\equiv_\infty [Q]
  \quad\Longleftrightarrow\quad
  \RunMean_{\CCS}([P])=\RunMean_{\CCS}([Q]).
\]
The flattened version is
\[
  [P]\equiv^\flat_\infty [Q]
  \quad\Longleftrightarrow\quad
  \Res_\beta([P])=\Res_\beta([Q]).
\]
\end{definition}

The first law says that an isolated initial action prefix has no residual-tail content.  The action may be visible or silent.  This does not make the prefix behaviourally invisible in CCS; it only says that, after the prefix has been consumed and the infinite tail is being measured, the initial one-step difference has disappeared.

\begin{proposition}[Prefix erasure]
\label{prop:ccs-prefix-erasure}
Let \(\alpha\in\Lambda\) and let \(P\) be a closed guarded CCS process.  Then
\[
  \RunMean_{\CCS}([\alpha.P])
  =
  \RunMean_{\CCS}([P]),
\]
and hence
\[
  \Res_\beta([\alpha.P])
  =
  \Res_\beta([P]).
\]
Equivalently,
\[
  [\alpha.P]\equiv_\infty [P]
  \qquad\text{and}\qquad
  [\alpha.P]\equiv^\flat_\infty [P].
\]
\end{proposition}

\begin{proof}
The process \(\alpha.P\) has the unique initial prefix transition
\[
  [\alpha.P]\xrightarrow{\alpha}[P].
\]
Thus every infinite residual run from \([\alpha.P]\) has the form
\[
  [\alpha.P],\ [P],\ [P_1],\ [P_2],\ldots,
\]
where
\[
  [P],\ [P_1],\ [P_2],\ldots
\]
is an infinite residual run from \([P]\).  Conversely, every infinite residual run from \([P]\) gives an infinite residual run from \([\alpha.P]\) by adding this initial prefix step.

The two corresponding streams agree after deleting finitely many initial terms.  By \Cref{cor:tail-invariance}, they have the same Stone--\v{C}ech residual meaning.  This gives equality of the run-sensitive sets of meanings.  Equality of the flattened residual sets follows by taking unions of the same compact meanings.
\end{proof}

Iterating the preceding proposition removes any finite string of initial prefixes from the residual-tail semantics.  If \(w=\alpha_1\cdots\alpha_k\in\Lambda^*\), write
\[
  w.P=\alpha_1.\alpha_2.\cdots\alpha_k.P,
\]
with the convention that the empty word gives \(P\).

\begin{corollary}[Finite prefix erasure]
\label{cor:ccs-finite-prefix-erasure}
For every finite word \(w\in\Lambda^*\) and every closed guarded process \(P\),
\[
  \RunMean_{\CCS}([w.P])
  =
  \RunMean_{\CCS}([P])
\]
and
\[
  \Res_\beta([w.P])
  =
  \Res_\beta([P]).
\]
\end{corollary}

\begin{proof}
Induct on the length of \(w\), using \Cref{prop:ccs-prefix-erasure} at each step.
\end{proof}

The same tail argument explains why guarded constants may be unfolded without changing the asymptotic residual meaning.  This is not an additional structural congruence assumption.  It is a consequence of the standard transition rule for constants: after the first unfolding step, the possible successor residuals are the same.

\begin{proposition}[Guarded unfolding]
\label{prop:ccs-guarded-unfolding}
Suppose \(A\) is a process constant with guarded defining equation
\[
  A\stackrel{\mathrm{def}}{=}P_A.
\]
Then
\[
  \RunMean_{\CCS}([A])
  =
  \RunMean_{\CCS}([P_A]),
\]
and hence
\[
  \Res_\beta([A])
  =
  \Res_\beta([P_A]).
\]
\end{proposition}

\begin{proof}
By the standard unfolding rule for constants, \([A]\xrightarrow{\mu}[R]\) exactly when \([P_A]\xrightarrow{\mu}[R]\).  Hence an infinite residual run from \([A]\) is obtained from an infinite residual run from \([P_A]\) by replacing the initial term \([P_A]\) with \([A]\), and conversely.  The resulting streams agree from the first successor state onward.  Tail invariance, \Cref{cor:tail-invariance}, therefore gives the same Stone--\v{C}ech residual meaning for corresponding runs.  Taking all runs gives equality of \(\RunMean_{\CCS}\), and taking unions gives equality of \(\Res_\beta\).
\end{proof}

Taken together, prefix erasure and guarded unfolding express a narrow but useful invariance principle.  They are not trace laws, nor are they bisimulation laws.  They say that the residual meaning attached to an infinite execution depends only on the eventual residual stream, so a finite initial segment contributes no point to the compact tail set.

\subsection{Choice as union and prefix-choice normalisation}
\label{subsec:ccs-choice-normalisation}

The preceding tail laws remove finite prefixes once a run has selected a tail.  Finite external choice has a similarly direct residual-tail reading.  An infinite run from \(P+Q\) makes its first move by using a transition of one summand, and after that point it is simply following a run of that summand.  Thus choice is interpreted by union of possible tail meanings.  Again, this is not a trace equation: it is a statement about the compact residual sets obtained after commitment to an infinite tail.

\begin{proposition}[Choice as union]
\label{prop:ccs-choice-union}
Let \(P\) and \(Q\) be closed guarded CCS processes.  Then
\[
  \RunMean_{\CCS}([P+Q])
  =
  \RunMean_{\CCS}([P])\cup\RunMean_{\CCS}([Q]),
\]
and hence
\[
  \Res_\beta([P+Q])
  =
  \Res_\beta([P])\cup\Res_\beta([Q]).
\]
\end{proposition}

\begin{proof}
Let
\[
  \rho=([R_0],[R_1],[R_2],\ldots)
\]
be an infinite residual run from \([P+Q]\).  Thus \([R_0]=[P+Q]\).  By the standard CCS rules for choice, the first transition of \(\rho\) is inherited from one of the two summands.  Hence either there is a transition \([P]\xrightarrow{\mu}[R_1]\), or there is a transition \([Q]\xrightarrow{\mu}[R_1]\).  In the first case
\[
  ([P],[R_1],[R_2],\ldots)
\]
is an infinite residual run from \([P]\); in the second case the analogous sequence is an infinite residual run from \([Q]\).  In either case the new run and \(\rho\) agree from the first successor state onward, so their Stone--\v{C}ech residual meanings are equal by \Cref{cor:tail-invariance}.  This proves
\[
  \RunMean_{\CCS}([P+Q])
  \subseteq
  \RunMean_{\CCS}([P])\cup\RunMean_{\CCS}([Q]).
\]

Conversely, suppose \(K\in\RunMean_{\CCS}([P])\).  Choose an infinite residual run
\[
  \sigma=([P],[P_1],[P_2],\ldots)
\]
from \([P]\) with \(\M_{\CCS}(\sigma)=K\).  Its first transition is some \([P]\xrightarrow{\mu}[P_1]\).  By the choice rule, \([P+Q]\xrightarrow{\mu}[P_1]\), so
\[
  ([P+Q],[P_1],[P_2],\ldots)
\]
is an infinite residual run from \([P+Q]\).  It agrees with \(\sigma\) after finitely many initial terms, and so has the same residual meaning.  Thus \(K\in\RunMean_{\CCS}([P+Q])\).  The same argument applies to \(K\in\RunMean_{\CCS}([Q])\).  This gives the reverse inclusion.

The equality for flattened residual sets follows by taking unions of the compact meanings occurring in the run-sensitive semantics.
\end{proof}

By induction, finite sums are interpreted by finite unions of run-sensitive meanings.  If
\[
  P_1+\cdots+P_m
\]
is any fixed bracketing of a finite choice, then
\[
  \RunMean_{\CCS}([P_1+\cdots+P_m])
  =
  \bigcup_{i=1}^m\RunMean_{\CCS}([P_i]),
\]
and similarly for \(\Res_\beta\).  The statement is independent of the bracketing because the right hand side is.

Combining prefix erasure with choice gives a small normalisation theorem for asymptotic residual semantics.  The normal form is shallow: it removes only finite prefix-choice structure sitting above a finite family of leaves.  The leaves may themselves contain recursion, parallel composition, restriction, relabelling, or further CCS structure.  No claim is made that those inner processes can be simplified by the same argument.

\begin{definition}[Prefix-choice wrappers]
\label{def:prefix-choice-wrapper}
Fix a set \(\mathcal L_0\) of closed guarded CCS processes, called leaves.  A finite prefix-choice wrapper over \(\mathcal L_0\) is generated by
\[
  E ::= R \mid \alpha.E \mid E+E,
\]
where \(R\in\mathcal L_0\) and \(\alpha\in\Lambda\).  The finite set of leaves of \(E\), written \(\Leaves(E)\), is defined recursively by
\[
\begin{aligned}
  \Leaves(R) &= \{R\},\\
  \Leaves(\alpha.E) &= \Leaves(E),\\
  \Leaves(E_1+E_2) &= \Leaves(E_1)\cup\Leaves(E_2).
\end{aligned}
\]
\end{definition}

\begin{theorem}[Prefix-choice normalisation]
\label{thm:ccs-prefix-choice-normalisation}
Let \(E\) be a finite prefix-choice wrapper over a set of closed guarded CCS processes.  Then
\[
  \RunMean_{\CCS}([E])
  =
  \bigcup_{R\in\Leaves(E)}\RunMean_{\CCS}([R]),
\]
and
\[
  \Res_\beta([E])
  =
  \bigcup_{R\in\Leaves(E)}\Res_\beta([R]).
\]
\end{theorem}

\begin{proof}
Induct on the construction of \(E\).  If \(E=R\) is a leaf, the statement is immediate.  If \(E=\alpha.E'\), then \Cref{prop:ccs-prefix-erasure} gives
\[
  \RunMean_{\CCS}([\alpha.E'])=\RunMean_{\CCS}([E']),
\]
and the induction hypothesis applies to \(E'\).  The flattened equality is the same argument using the flattened part of \Cref{prop:ccs-prefix-erasure}.

If \(E=E_1+E_2\), then \Cref{prop:ccs-choice-union} gives
\[
  \RunMean_{\CCS}([E_1+E_2])
  =
  \RunMean_{\CCS}([E_1])\cup\RunMean_{\CCS}([E_2]).
\]
Applying the induction hypothesis to \(E_1\) and \(E_2\) gives the required union over
\(\Leaves(E_1)\cup\Leaves(E_2)=\Leaves(E)\).  The flattened equality follows in the same way.
\end{proof}

For example,
\[
  \RunMean_{\CCS}([a.(b.P+c.Q)+\tau.d.R])
  =
  \RunMean_{\CCS}([P])
  \cup
  \RunMean_{\CCS}([Q])
  \cup
  \RunMean_{\CCS}([R]).
\]
Thus finite prefixes and finite choices can be collapsed before the asymptotic residual meaning is computed, provided they occur in this outer prefix-choice form.  The displayed equality is a semantic normalisation for residual tails.  It does not say that the original process and the displayed leaves have the same labelled transition behaviour, nor that the normalisation may be pushed through arbitrary CCS contexts.

\subsection{The parallel-composition boundary}
\label{subsec:ccs-parallel-different}

The finite prefix-choice theorem is an outermost normalisation theorem.  It is tempting to read it as the beginning of a quotient semantics for all CCS terms, with \(\equiv_\infty\) used as a congruence and parallel composition interpreted directly on residual meanings.  The next example rules out that reading.  A prefix which disappears from the tail of a component may still be used before it disappears: in a parallel context it can synchronise with the environment and thereby change which infinite residual runs exist.

The boundary is already visible in a small closed context.

\begin{example}[Prefix erasure is not a full CCS congruence]
\label{ex:ccs-prefix-not-congruence}
Let \(a\) and \(b\) be distinct visible action names, and let
\[
  L \stackrel{\mathrm{def}}{=} b.L.
\]
By \Cref{prop:ccs-prefix-erasure},
\[
  [a.0]\equiv_\infty [0]
  \qquad\text{and}\qquad
  [a.0]\equiv^\flat_\infty [0],
\]
since neither \([a.0]\) nor \([0]\) has an infinite residual run.

Now place the two processes in the CCS context
\[
  C[-]=((- )\mid \overline a.L)\setminus\{a\}.
\]
In \(C[a.0]\), the two complementary prefixes can synchronise, giving a silent transition
\[
  [(a.0\mid \overline a.L)\setminus\{a\}]
  \xrightarrow{\tau}
  [(0\mid L)\setminus\{a\}].
\]
Since \(a\) and \(b\) are distinct, the residual process \([(0\mid L)\setminus\{a\}]\) has an infinite \(b\)-run.  Hence \(C[a.0]\) has an infinite residual run whose tail is the stable loop on \([(0\mid L)\setminus\{a\}]\).  In particular,
\[
  \RunMean_{\CCS}(C[a.0])\neq\varnothing.
\]

By contrast, in \(C[0]=[(0\mid\overline a.L)\setminus\{a\}]\), the visible \(\overline a\)-transition is blocked by the restriction and there is no complementary \(a\)-transition with which it can synchronise.  Thus \(C[0]\) has no infinite residual run, and
\[
  \RunMean_{\CCS}(C[0])=\varnothing.
\]
The same example also gives
\[
  C[a.0]\not\equiv^\flat_\infty C[0],
\]
since the flattened residual set of \(C[a.0]\) is nonempty while that of \(C[0]\) is empty.
\end{example}

\begin{proposition}[Prefix erasure is not contextual]
\label{prop:ccs-prefix-erasure-not-contextual}
The equivalences \(\equiv_\infty\) and \(\equiv^\flat_\infty\) are not congruences for arbitrary CCS contexts.
\end{proposition}

\begin{proof}
The processes \([a.0]\) and \([0]\) are equivalent for both relations, as noted in \Cref{ex:ccs-prefix-not-congruence}.  The context
\[
  C[-]=((- )\mid\overline a.L)\setminus\{a\}
\]
separates them: \(C[a.0]\) has an infinite residual run after the initial synchronisation, whereas \(C[0]\) has none.  Hence both the run-sensitive and flattened residual semantics distinguish the two contextual instances.
\end{proof}

The reason is operational.  The residual tail of a composite process may depend on synchronisation possibilities that are invisible in the separate asymptotic meanings of its components.  The prefix in \(a.0\) is absent from every infinite tail of the isolated process, but in the context above it is consumed before the tail begins and thereby exposes the looping process \(L\).

Consequently, one should not expect a binary operation
\[
  \K(\beta\CCS)\times\K(\beta\CCS)
  \longrightarrow
  \K(\beta\CCS)
\]
which computes the run-sensitive meaning of \(P\mid Q\) from the run-sensitive meanings of \(P\) and \(Q\) alone.  Extra operational information is needed: at least the enabled initial actions relevant to synchronisation, and in general the way the two residual streams are correlated by a common schedule.  This is why the relational construction of \Cref{sec:relational} compactifies tuples of observations along the same asymptotic view of time, rather than replacing joint behaviour by a product of separate meanings.

For the rest of this paper, the positive algebraic laws are therefore used only for finite prefix-choice structure.  Parallel composition remains part of the residual transition system whose infinite runs are collected; it is not collapsed by the prefix-choice normalisation theorem.

\subsection{The discrete process topology}

In the discrete topology every subset of \(\CCS\) is clopen.  Principal points of \(\beta\CCS\) are actual structural-congruence classes of CCS processes.  Non-principal points of \(\beta\CCS\setminus\CCS\) represent asymptotic escape through infinite families of residual processes.

\begin{proposition}[Residual recurrence in the discrete topology]
\label{prop:ccs-residual-recurrence}
Let
\[
  \rho=([P_0],[P_1],[P_2],\ldots)
\]
be an infinite residual run in the discrete space \(\CCS\).  Then:
\begin{enumerate}[label=(\roman*),leftmargin=2.5em]
\item
For \([Q]\in\CCS\),
\[
  [Q]\in\M_{\CCS}(\rho)
  \quad\Longleftrightarrow\quad
  [P_n]=[Q]\text{ for infinitely many }n.
\]
\item
For \(C\subseteq\CCS\),
\[
  \M_{\CCS}(\rho)\subseteq\widehat C
  \quad\Longleftrightarrow\quad
  [P_n]\in C\text{ eventually}.
\]
\item
For \(C\subseteq\CCS\),
\[
  \M_{\CCS}(\rho)\cap\widehat C\neq\varnothing
  \quad\Longleftrightarrow\quad
  [P_n]\in C\text{ for infinitely many }n.
\]
\end{enumerate}
\end{proposition}

\begin{proof}
This is \Cref{thm:discrete-temporal} with \(T=\N\), the cofinite filter, \(X=\CCS\), and \(a=\rho\).  For (i), take \(C=\{[Q]\}\).  Since \([Q]\) is a principal point of \(\beta\CCS\), the clopen set \(\widehat{\{[Q]\}}\) is the singleton \(\{[Q]\}\).
\end{proof}

\subsection{Stable and recurrent divergence}

\begin{example}[Stable divergence]
\label{ex:ccs-stable-divergence}
Let
\[
  A\stackrel{\mathrm{def}}{=}a.A.
\]
There is an infinite run
\[
  [A]\xrightarrow{a}[A]\xrightarrow{a}[A]\xrightarrow{a}\cdots.
\]
The associated residual stream is constant.  Hence
\[
  \M_{\CCS}(\rho)=\{[A]\}.
\]
\end{example}

\begin{example}[Finite recurrent divergence]
\label{ex:ccs-finite-recurrent-divergence}
Let
\[
  P\stackrel{\mathrm{def}}{=}a.Q,
  \qquad
  Q\stackrel{\mathrm{def}}{=}b.P.
\]
The evident infinite run alternates between \([P]\) and \([Q]\):
\[
  [P]\xrightarrow{a}[Q]\xrightarrow{b}[P]\xrightarrow{a}[Q]\xrightarrow{b}\cdots.
\]
Its residual meaning is
\[
  \M_{\CCS}(\rho)=\{[P],[Q]\}.
\]
\end{example}

Both meanings are contained in the principal part \(\CCS\subseteq\beta\CCS\).  The first is a singleton; the second is a finite recurrent set.

\subsection{Unbounded residual growth}

\begin{example}[Growing residuals]
\label{ex:ccs-growing-residuals}
Let
\[
  G\stackrel{\mathrm{def}}{=}a.(G\mid G).
\]
There is an infinite run whose \(n\)-th residual is, up to associativity, commutativity and unit laws for \(\mid\), a parallel composition of \(n+1\) copies of \(G\):
\[
  [G]
  \xrightarrow{a}
  [G\mid G]
  \xrightarrow{a}
  [G\mid G\mid G]
  \xrightarrow{a}
  \cdots.
\]
\end{example}

\begin{proposition}[Escape of the growing run]
\label{prop:ccs-growing-escape}
Let \(\rho_G\) be the run in \Cref{ex:ccs-growing-residuals}.  Then
\[
  \M_{\CCS}(\rho_G)\cap\CCS=\varnothing.
\]
Equivalently,
\[
  \M_{\CCS}(\rho_G)\subseteq\beta\CCS\setminus\CCS.
\]
\end{proposition}

\begin{proof}
Each transition fires one copy of \(G\), replacing it by two copies.  Thus the number of parallel copies increases by one at each step.  Since the process classes are taken modulo the usual structural congruence for parallel composition, the residual after \(n\) steps is the class of a parallel composition of \(n+1\) copies of \(G\).  These classes are pairwise distinct.

By \Cref{prop:ccs-residual-recurrence}(i), a principal point \([Q]\in\CCS\) belongs to \(\M_{\CCS}(\rho_G)\) exactly when \([Q]\) occurs infinitely often along the run.  No residual class occurs more than once.  Hence no principal point lies in the meaning.
\end{proof}

\begin{corollary}[Separation of stable and growing divergence]
\label{cor:ccs-stable-growing-separation}
Let
\[
  L\stackrel{\mathrm{def}}{=}a.L
  \qquad\text{and}\qquad
  G\stackrel{\mathrm{def}}{=}a.(G\mid G).
\]
The evident infinite run of \(L\) has residual meaning \(\{[L]\}\), while the evident growing run of \(G\) has residual meaning contained in \(\beta\CCS\setminus\CCS\).  Thus the two runs are set-theoretically separated by the principal part \(\CCS\subseteq\beta\CCS\).
\end{corollary}

\begin{proof}
The run of \(L\) is the stable-divergence example.  The run of \(G\) is covered by \Cref{prop:ccs-growing-escape}.
\end{proof}

\subsection{Escape detected by unbounded observables}
\label{subsec:ccs-unbounded-observables}

The previous example detects escape by showing that no residual process recurs exactly.  Often one wants a more concrete explanation, expressed through an ordinary observable rather than through the set-theoretic principal part of \(\beta\CCS\).  The following proposition says that escape can be certified by an integer-valued observation, such as size, number of parallel components, number of pending messages, or another resource measure.  Individual points of \(\beta\CCS\setminus\CCS\) need not be inspected; it is enough to look at their image under such observations.

\begin{definition}[Resource observable]
A resource observable on the discrete process space \(\CCS\) is a function
\[
  r:\CCS\to\N,
\]
where \(\N\) is given the discrete topology.  For \(k\in\N\), write
\[
  B_{\leq k}^r
  =
  \{[P]\in\CCS:r([P])\leq k\}
\]
for the region where the resource is bounded by \(k\).
\end{definition}

\begin{proposition}[Escape detected by an unbounded observable]
\label{prop:ccs-unbounded-observable-escape}
Let
\[
  \rho=([P_0],[P_1],[P_2],\ldots)
\]
be an infinite residual run in the discrete process space \(\CCS\), and let \(r:\CCS\to\N\) be a resource observable.  Suppose that
\[
  r([P_n])\to\infty
\]
in the sense that, for every \(k\in\N\), there is \(N\) such that \(r([P_n])>k\) for all \(n\geq N\).  Then
\[
  \beta r[\M_{\CCS}(\rho)]
  =
  \M_{\N}(r\circ\rho)
  \subseteq
  \N^*.
\]
Moreover, for every \(k\in\N\),
\[
  \M_{\CCS}(\rho)
  \cap
  \cl_{\beta\CCS} B_{\leq k}^r
  =
  \varnothing.
\]
In particular,
\[
  \M_{\CCS}(\rho)\cap\CCS=\varnothing.
\]
\end{proposition}

\begin{proof}
Since \(\CCS\) is discrete, the map \(r:\CCS\to\N\) is continuous.  By functoriality,
\[
  \M_{\N}(r\circ\rho)
  =
  \beta r[\M_{\CCS}(\rho)].
\]
The stream \(r\circ\rho:\N\to\N\) eventually avoids every finite subset of \(\N\).  Hence no principal point \(m\in\N\) belongs to \(\M_{\N}(r\circ\rho)\): by \Cref{thm:discrete-temporal}, this would mean that the value \(m\) occurs infinitely often.  Thus
\[
  \M_{\N}(r\circ\rho)\subseteq\N^*.
\]

Now fix \(k\).  If
\[
  p\in
  \M_{\CCS}(\rho)
  \cap
  \cl_{\beta\CCS}B_{\leq k}^r,
\]
then continuity of \(\beta r\) gives
\[
  \beta r(p)
  \in
  \cl_{\beta\N}r[B_{\leq k}^r]
  \subseteq
  \{0,1,\ldots,k\}.
\]
This contradicts \(\beta r(p)\in\N^*\).  Therefore the displayed intersection is empty for every \(k\).

Finally, if a principal process point \([Q]\in\CCS\) belonged to \(\M_{\CCS}(\rho)\), then with \(k=r([Q])\) it would belong to \(B_{\leq k}^r\), and hence to \(\cl_{\beta\CCS}B_{\leq k}^r\), contradicting the previous paragraph.
\end{proof}

\begin{remark}
The converse is not asserted.  A run may escape from the principal part because it visits infinitely many different residual processes, even though a particular chosen resource observable stays bounded.  The proposition is useful when a concrete unbounded quantity explains the escape.
\end{remark}

\begin{corollary}[The growing CCS process escapes by an unbounded resource]
\label{cor:ccs-growing-unbounded-observable}
Let \(G\) and \(\rho_G\) be as in \Cref{ex:ccs-growing-residuals}.  For \(m\geq1\), write
\[
  G^{(m)}
  =
  \underbrace{G\mid\cdots\mid G}_{m\text{ copies}},
\]
modulo the associativity, commutativity and unit laws for parallel composition.  Define
\[
  r_G([P])
  =
  \begin{cases}
    m, & \text{if }[P]=[G^{(m)}]\text{ for some }m\geq1,\\
    0, & \text{otherwise.}
  \end{cases}
\]
Then \(r_G\) is a resource observable and, along the growing run,
\[
  r_G(\rho_G(n))=n+1.
\]
Consequently,
\[
  \beta r_G[\M_{\CCS}(\rho_G)]\subseteq\N^*,
\]
and \(\M_{\CCS}(\rho_G)\) is disjoint from the closure of every region in which the number of parallel \(G\)-components is bounded.
\end{corollary}

\begin{proof}
The residual after \(n\) transitions in the growing run is \([G^{(n+1)}]\).  Thus \(r_G(\rho_G(n))=n+1\), which tends to infinity.  The conclusion follows from \Cref{prop:ccs-unbounded-observable-escape}.
\end{proof}

\subsection{Nondeterministic residual meanings}

\begin{example}[Branching into stable and growing behaviour]
\label{ex:ccs-branching-stable-growing}
Let
\[
  B\stackrel{\mathrm{def}}{=}a.L + a.G,
\]
where
\[
  L\stackrel{\mathrm{def}}{=}a.L,
  \qquad
  G\stackrel{\mathrm{def}}{=}a.(G\mid G).
\]
From \([B]\) there is an \(a\)-transition to \([L]\) and an \(a\)-transition to \([G]\).  Thus \(\RunMean_{\CCS}([B])\) contains the stable meaning \(\{[L]\}\) and also contains the non-principal compact meaning arising from the growing run of \(G\).
\end{example}

This example illustrates the difference between run-sensitive and flattened meanings.  The run-sensitive semantics records that the stable and growing behaviours arise from different executions.  The flattened residual set records the union of the asymptotic points that may arise.

\subsection{Observation topologies and modal observations}
\label{subsec:ccs-observation-topologies}

The discrete topology distinguishes all structural-congruence classes.  Coarser process observations may be obtained from a Boolean algebra of observable process properties.

Let \(\mathcal B\subseteq\pow(\CCS)\) be a Boolean algebra of subsets of \(\CCS\).  The topology generated by \(\mathcal B\) as a clopen base is zero-dimensional.  If \(\mathcal B\) separates points, this topology is Hausdorff and hence Tychonoff.  If \(\mathcal B\) does not separate points, replace \(\CCS\) by the quotient
\[
  \CCS/{\equiv_{\mathcal B}},
\]
where
\[
  [P]\equiv_{\mathcal B}[Q]
  \quad\Longleftrightarrow\quad
  \forall B\in\mathcal B\
  ([P]\in B\Leftrightarrow [Q]\in B).
\]
The induced Boolean algebra separates the quotient points, so the quotient carries a zero-dimensional Hausdorff topology.  In the rest of this subsection, \(\CCS_{\mathcal B}\) denotes either the separated observation space itself or this quotient.

A standard source of observations is a modal process logic.  Let \(\mathcal L\) be a Hennessy--Milner style modal language interpreted over CCS processes \cite{HennessyMilner1980,HennessyMilner1985}, and let
\[
  [\![\varphi]\!]
  =
  \{[P]\in\CCS:[P]\models\varphi\}
\]
be the satisfaction set of a formula \(\varphi\).  A Boolean algebra generated by such satisfaction sets gives an observation topology as above.

For an infinite residual run \(\rho=([P_0],[P_1],[P_2],\ldots)\), write
\[
  q_{\mathcal B}:\CCS\to\CCS_{\mathcal B}
\]
for the observation quotient map, and consider the observed residual run
\[
  q_{\mathcal B}\circ\rho:\N\to\CCS_{\mathcal B}.
\]

\begin{proposition}[Modal eventuality and recurrence]
\label{prop:ccs-modal-eventuality}
Let \(\varphi\) be a process formula whose satisfaction set belongs to the observation algebra \(\mathcal B\).  Let \(C_\varphi\subseteq\CCS_{\mathcal B}\) be the clopen set corresponding to \([\![\varphi]\!]\).  Then
\[
  \M_{\CCS_{\mathcal B}}(q_{\mathcal B}\circ\rho)
  \subseteq
  \cl_{\beta\CCS_{\mathcal B}}C_\varphi
\]
if and only if \([P_n]\models\varphi\) eventually.  Also,
\[
  \M_{\CCS_{\mathcal B}}(q_{\mathcal B}\circ\rho)
  \cap
  \cl_{\beta\CCS_{\mathcal B}}C_\varphi
  \neq\varnothing
\]
if and only if \([P_n]\models\varphi\) for infinitely many \(n\).
\end{proposition}

\begin{proof}
The set \(C_\varphi\) is clopen in \(\CCS_{\mathcal B}\) by construction of the observation topology.  Apply \Cref{prop:clopen-temporal} to the stream \(q_{\mathcal B}\circ\rho\) and to the clopen set \(C_\varphi\), using the cofinite filter on \(\N\).
\end{proof}

\begin{example}[Persistent enabledness]
Let
\[
  \varphi=\langle a\rangle\top.
\]
Then \([P]\models\varphi\) says that the residual process \([P]\) has an outgoing \(a\)-transition.  For a residual run \(([P_n])_{n<\omega}\), eventual inclusion in the clopen region for \(\varphi\) says that, from some point onward, every residual process has an \(a\)-transition.  Nonempty intersection with the same clopen region says that \(a\)-enabled residuals occur infinitely often.
\end{example}

\subsection{Relational comparison of residual runs}

Let \(R\subseteq\CCS\times\CCS\) be a relation on residual processes.  Examples include equality of structural-congruence classes, strong bisimilarity, weak bisimilarity, simulation, or agreement on a chosen observation algebra.  Given two residual runs
\[
  \rho=([P_n])_{n<\omega},
  \qquad
  \sigma=([Q_n])_{n<\omega},
\]
their paired run is
\[
  \langle\rho,\sigma\rangle:\N\to\CCS\times\CCS,
  \qquad
  n\mapsto([P_n],[Q_n]).
\]

\begin{proposition}[Asymptotic relational matching]
\label{prop:ccs-relational-matching}
Assume \(\CCS\) is given the discrete topology.  Then
\[
  \M_{\CCS\times\CCS}(\langle\rho,\sigma\rangle)
  \subseteq
  \widehat R
\]
if and only if \(([P_n],[Q_n])\in R\) eventually.  Also,
\[
  \M_{\CCS\times\CCS}(\langle\rho,\sigma\rangle)
  \cap
  \widehat R
  \neq\varnothing
\]
if and only if \(([P_n],[Q_n])\in R\) for infinitely many \(n\).
\end{proposition}

\begin{proof}
Apply \Cref{thm:discrete-temporal} to the discrete output space \(\CCS\times\CCS\) and to the stream \(n\mapsto([P_n],[Q_n])\).
\end{proof}

The compactification used here is \(\beta(\CCS\times\CCS)\).  It records the relation between the two residual streams along the same ultrafilter on time.  The separate compact meanings \(\M_{\CCS}(\rho)\) and \(\M_{\CCS}(\sigma)\) do not, in general, determine this paired meaning.

\subsection{Collected CCS consequences}
\label{subsec:ccs-collected-consequences}

The preceding results give the following compact summary of the CCS application.  It collects the residual-tail information obtained from the general Stone--\v{C}ech construction, the finite abstraction results, and the process-specific prefix-choice laws.

\begin{proposition}[CCS residual-tail consequences]
\label{prop:ccs-residual-tail-consequences}
Let
\[
  \rho=([P_0],[P_1],[P_2],\ldots)
\]
be an infinite residual run in the discrete process space \(\CCS\).  Then the following hold.
\begin{enumerate}[label=(\roman*),leftmargin=2.5em]
\item
If \(\rho\) is eventually constant with value \([Q]\), then
\[
  \M_{\CCS}(\rho)=\{[Q]\}.
\]
If \(\rho\) is eventually periodic and its periodic tail visits exactly the finite set
\(S\subseteq\CCS\), then
\[
  \M_{\CCS}(\rho)=S.
\]

\item
More generally, the principal part of \(\M_{\CCS}(\rho)\) consists exactly of the residual process classes which occur infinitely often along \(\rho\).  For every \(C\subseteq\CCS\), containment of \(\M_{\CCS}(\rho)\) in \(\widehat C\) is equivalent to eventual membership of the run in \(C\), and nonempty intersection with \(\widehat C\) is equivalent to infinitely-often membership in \(C\).

\item
If \(r:\CCS\to\N\) is a resource observable and \(r([P_n])\to\infty\), then
\[
  \beta r[\M_{\CCS}(\rho)]
  =
  \M_\N(r\circ\rho)
  \subseteq
  \N^*,
\]
and \(\M_{\CCS}(\rho)\cap\CCS=\varnothing\).  Thus unbounded residual growth is visible through an integer-valued observation whenever such a resource map is available.

\item
Let \(q_{\mathcal A}:\CCS\to\CCS_{\mathcal A}\) be a finite observational quotient generated by finitely many clopen process predicates.  Then
\[
  \M_{\CCS_{\mathcal A}}(q_{\mathcal A}\circ\rho)
  =
  \beta q_{\mathcal A}[\M_{\CCS}(\rho)],
\]
and this finite meaning is exactly the set of abstract residual states visited infinitely often by \(q_{\mathcal A}\circ\rho\).

\item
If the observed run \(q_{\mathcal A}\circ\rho\) is edge-fair in the sense of \Cref{prop:edge-fair-finite-graph} as an infinite path in a finite abstract transition graph, then its finite meaning is a reachable terminal strongly connected component of that graph.  Conversely, every reachable terminal strongly connected component which supports an infinite path is the meaning of some edge-fair infinite path, with the usual convention that deadlocked terminal states may be given stuttering self-loops.

\item
For every finite prefix-choice wrapper \(E\),
\[
  \RunMean_{\CCS}([E])
  =
  \bigcup_{R\in\Leaves(E)}\RunMean_{\CCS}([R])
\]
and the analogous equality holds for \(\Res_\beta\).  These equalities are residual-tail laws for outer prefix-choice structure.  They are not congruence laws for arbitrary CCS contexts: parallel composition and synchronisation can distinguish processes with the same isolated residual-tail semantics.
\end{enumerate}
\end{proposition}

\begin{proof}
Part (i) is the finite discrete case of \Cref{prop:ccs-residual-recurrence}; the eventually periodic case says that the values in the period are exactly the values occurring infinitely often.  Part (ii) is precisely \Cref{prop:ccs-residual-recurrence}.  Part (iii) is \Cref{prop:ccs-unbounded-observable-escape}.  Part (iv) is \Cref{prop:finite-observational-abstraction}, applied to the quotient map \(q_{\mathcal A}\).  Part (v) is \Cref{prop:edge-fair-finite-graph}.  Part (vi) combines \Cref{thm:ccs-prefix-choice-normalisation} with \Cref{prop:ccs-prefix-erasure-not-contextual}.
\end{proof}

\section{Mobile calculi}
\label{sec:mobile}

The CCS instance above uses residual processes in a fixed process universe.  Mobile calculi add further structure.  In the \(\pi\)-calculus, residual processes involve binding, \(\alpha\)-conversion, restriction, fresh names, and scope extrusion.  A Stone--\v{C}ech residual semantics can be formulated in the same general way once the process space has been chosen, but the choice is more delicate than for CCS.  Standard references for the \(\pi\)-calculus include Milner's book and the book of Sangiorgi and Walker \cite{Milner1999,SangiorgiWalker2001}.

One natural presentation would take \(\pi\)-processes modulo \(\alpha\)-conversion and structural congruence, with observations invariant under finite permutations of names.  Nominal sets provide a general mathematical setting for names and binding \cite{Pitts2013}.  The resulting topology should either separate the chosen quotient points or be followed by the observational quotient described earlier.  Residual runs would then be streams
\[
  P_0\to P_1\to P_2\to\cdots
\]
in this quotient process space, and their meanings would again be compact subsets of the corresponding Stone--\v{C}ech compactification.

Escape in the residual process space may now arise not only from syntactic growth or increasing parallel structure, but also from unbounded fresh-name generation, changing communication topology, or repeated scope extrusion.  For example, an observation measuring the number of active restricted names, formulated invariantly under \(\alpha\)-conversion, would give finite or infinite quotients detecting whether such name structure remains bounded, recurs, or grows without bound along a run.

A full treatment would require fixing a particular \(\pi\)-calculus syntax, transition system, structural congruence, name-invariant observation algebra, and treatment of administrative or weak transitions.  The CCS case supplies the worked instance of the residual-process construction.

\section{Related work and comparison}
\label{sec:related-work}

The ingredients of the construction are mostly classical.  Stone--\v{C}ech compactification, filters, tail cluster sets, operational transition systems and finite quotient graphs all have long histories.  The particular use made of them here is more specific: residual executions are treated as streams of observations, and their tails are collected in a compact space before finite, modal or resource-level observations are applied.  This places the paper near several bodies of work, but not squarely inside any one of them.

\subsection{Compactified cluster sets and topological dynamics}

For a sequence or net in a compact space, the set of tail accumulation points is a standard object.  In topological dynamics the corresponding construction is the \(\omega\)-limit set of an orbit \cite{Ellis1969}.  The examples in this paper differ from the usual compact dynamical setting because the natural observation spaces need not be compact.  A residual process may grow through larger and larger terms without returning to any exact residual state.  The CCS process
\[
  G=a.(G\mid G)
\]
is the simplest example used here: its residual stream escapes through larger parallel compositions, while finite observational quotients and integer-valued resource observations still give meaningful tail information.

The Stone--\v{C}ech compactification is the largest compactification of a Tychonoff space in the usual extension sense \cite{GillmanJerison1976,Engelking1989}.  The paper uses that universal property in a semantic way.  Every continuous observation \(f:X\to Y\) extends to \(\beta f:\beta X\to\beta Y\), and the compactified meaning satisfies
\[
  \M_Y^\Filt(f\circ a)=\beta f[\M_X^\Filt(a)].
\]
Thus finite quotients, modal-observation quotients and resource maps are not separate add-ons; they are continuous images of the same compact tail object.  Recent work on duality and codensity monads gives a useful broader setting for this kind of topological-categorical semantics, especially around profinite, Vietoris and ultrafilter-like constructions \cite{GehrkePetrisanReggio2020,LenkeMiliusUrbatWittrock2026}.  The present paper does not use codensity machinery, but those references help locate the Stone--\v{C}ech construction in a current rather than purely historical landscape.

Stone--\v{C}ech methods also play a major role in topological algebra, especially for discrete semigroups, where algebraic structure extends to \(\beta S\) and can be used to study recurrence and large combinatorial sets \cite{HindmanStrauss2012}.  Here the emphasis is different.  Compactification is used first to organise residual tail observations.  Algebraic structure on compactified residual spaces is left as a possible further question.

\subsection{Operational and process semantics}

Structural operational semantics describes programs by transitions between configurations \cite{Plotkin1981}.  CCS uses labelled transition systems to model communication, nondeterminism, synchronisation and behavioural equivalence \cite{Milner1989}.  The CCS part of this paper begins with that standard operational material, but it measures a different object.  An infinite execution
\[
  P_0\to P_1\to P_2\to\cdots
\]
is read as a stream of residual processes.

This separates the construction from trace semantics, bisimulation, testing equivalence and the usual modal characterisations of process equivalence.  Those semantics retain finite observations, branching structure, or both.  The residual-tail semantics records what remains after small sets of times have been discarded.  This is why finite prefixing, guarded unfolding and finite prefix-choice wrappers satisfy the tail laws proved in \Cref{sec:ccs}.  The same section gives the boundary of those laws: prefix erasure is not contextual for full CCS, since a finite prefix may synchronise with its environment before it disappears from the residual tail.

For calculi with binding, the situation is less immediate.  The construction can be repeated for a \(\pi\)-calculus process space once a quotient and a class of observations have been fixed, but names, binding, freshness and scope extrusion require additional structure.  Standard accounts of the \(\pi\)-calculus and nominal techniques provide the relevant background \cite{Milner1999,SangiorgiWalker2001,Pitts2013}.

\subsection{Temporal logic, fairness, and recurrence}

Temporal logics provide languages for specifying behaviour over time.  The safety/liveness distinction has a topological formulation in terms of sets of infinite behaviours \cite{AlpernSchneider1985}, and systems such as TLA and the Manna--Pnueli framework give expressive accounts of reactive computation \cite{Lamport1994,MannaPnueli1995}.

The present construction yields only a small temporal fragment directly, but that fragment is useful.  For a discrete observation space and \(A\subseteq X\),
\[
  \M_X(a)\subseteq\widehat A
\]
means eventual truth of \(A\), while
\[
  \M_X(a)\cap\widehat A\neq\varnothing
\]
means recurrent truth of \(A\).  The same reading applies to clopen observations in non-discrete spaces, using their clopen extensions in \(\beta X\).  More sequential properties, such as response obligations, bounded waiting and exact alternation patterns, require an observation space which already records the relevant obligations, clocks or histories.

The treatment of fairness by strengthening the time filter has a modest formal burden, but the word ``fairness'' covers many different assumptions in the process-algebra literature.  Recent work on progress, justness and fairness makes clear that these assumptions are not interchangeable \cite{vanGlabbeekHoefner2019}.  There is also recent work on expressing progress, justness and fairness assumptions in modal \(\mu\)-calculus formulae \cite{SpronckLuttikWillemse2024}.  The filter construction used here is compatible with that distinction: different fairness or progress requirements simply generate different filters, and hence different compact residual meanings.

\subsection{Powerdomains, collecting semantics, and abstraction}

Powerdomains model nondeterminism by replacing single outcomes with suitable collections of outcomes \cite{Plotkin1976,Smyth1978,AbramskyJung1994}.  Abstract interpretation similarly starts from a collecting semantics and then approximates it in a computable abstract domain \cite{CousotCousot1977}.  The present construction is a collecting semantics of a different kind.  Even one deterministic infinite run can have a non-singleton meaning: it may revisit several residual observations infinitely often, or combine a recurrent principal part with an escaping non-principal part.

Nondeterministic processes add a second layer of collection.  The run-sensitive semantics records one compact meaning for each infinite residual execution, while the flattened semantics takes the union of the resulting asymptotic points.  The finite-abstraction results give the computational interface to this compact semantics.  A finite Boolean algebra of observations gives a finite quotient, and in that quotient the meaning of a run is exactly the set of abstract states seen infinitely often.  For finite transition graphs this becomes strongly connected component analysis, and under the edge-fair reading of \Cref{prop:edge-fair-finite-graph} it becomes analysis of reachable terminal components.  These finite calculations are continuous finite images of the compact semantics.

\subsection{Coalgebraic and relational perspectives}

Coalgebraic semantics gives a general mathematical account of state-based systems and their observations, including transition systems and process behaviour \cite{Rutten2000}.  Coalgebraic trace semantics has continued to develop well beyond the early general theory; recent work such as \emph{Steps and Traces} gives a unified view of several approaches to coalgebraic trace semantics \cite{RotJacobsLevy2021}, and still more recent work studies compositionality for coalgebraic trace equivalence \cite{JourdeEtAl2026}.  The present paper is compatible with coalgebraic presentations of transition systems, but it is not a final-coalgebra or trace-equivalence semantics.  It starts from runs, or from runs generated by an operational transition system, and assigns compact tail meanings to the corresponding observation streams.

The relational construction in \Cref{sec:relational} addresses a different issue.  Some tail properties depend on correlations between observations made at the same time, or along the same asymptotic view of time.  Forming meanings in \(\beta(X_1\times\cdots\times X_k)\), rather than multiplying separate meanings afterwards, preserves these correlations.  This is relevant for comparing residual runs, recording agreement between components, and explaining why parallel composition cannot be reconstructed from separate unary meanings alone.

\subsection{Process observations and modal logics}

Hennessy--Milner logic relates modal observations of processes to behavioural equivalence in standard image-finite settings \cite{HennessyMilner1980,HennessyMilner1985}.  In the CCS application here, modal formulas have a more observational role.  Their satisfaction sets generate Boolean algebras of process predicates.  After separation, or after quotienting by observational indistinguishability, these Boolean algebras give Tychonoff process spaces.  The clopen temporal theorem then turns modal process observations into eventual and recurrent properties of residual executions.

Encodings between process calculi raise further questions, such as operational correspondence, divergence reflection, and preservation of behavioural equivalence.  General criteria for valid encodings are discussed by Gorla \cite{Gorla2010}.  For the present semantics, an encoding would also have to specify which residual observations are preserved and how the corresponding observation topologies are related.  Transporting compact residual meanings across calculi is therefore an additional semantic problem, rather than an automatic consequence of the general construction.

\section{Limitations and further directions}
\label{sec:limitations}

The construction records tail information.  This is the source of both its useful laws and its limitations.  If two residual streams agree after deleting a small set of times, they have the same compact meaning.  As a result, finite prefixes and finite prefix-choice wrappers disappear from the residual tail.  The non-contextuality proposition, \Cref{prop:ccs-prefix-erasure-not-contextual}, shows the matching boundary: a prefix can still matter before it disappears, because it may synchronise with an environment.  Tail equivalence is therefore not a congruence for arbitrary CCS contexts.

The semantics also depends on the chosen observation space.  The discrete topology on structural-congruence classes gives the finest version used in the paper.  Coarser Boolean observation algebras identify processes that satisfy the same chosen predicates, and different choices may induce different compact residual meanings.  This is not an anomaly; it is the usual dependence of an observational semantics on its observations.  The finite quotient results make the dependence explicit.

The compact carrier is not an effective data structure.  Free ultrafilters require choice, and spaces such as \(\beta\N\) are not finite or effectively presentable.  The computational content lies in images of the compact meaning.  Finite observational quotients, clopen predicates, relational observations and resource observables extract concrete information from it.  In particular, finite quotients reduce recurrent residual information to recurrent abstract states and strongly connected component calculations, while integer-valued resource observations detect escape by landing in \(\N^*\).

A related effective question is left open.  The definition of \(\M_X^\Filt(a)\) applies to arbitrary runs, not only to program-generated or computable runs.  If \(X\) is countable and discrete, there are only countably many computable streams \(\N\to X\), and hence only countably many exact compact meanings obtained from computable streams.  By contrast, \(\beta X\) has many more closed subsets.  Thus most closed sets in the compactified observation space cannot be exact meanings of computable streams.  At the same time, computable streams may still have meanings containing non-principal ultrafilter points: the identity stream \(n\mapsto n\) on \(\N\) has meaning \(\N^*\).  A systematic effective version would require additional computability structure on the observation space and on the chosen observations.

The treatment of time is also coarse.  It records tail membership, recurrence, escape and asymptotic correlation, but not the full order of a trace.  Bounded response times, exact alternation patterns, quantitative rates and protocol-style sequencing properties are not retained unless they have first been encoded into the observation space.  Similarly, the fairness construction specifies which asymptotic views of time are admissible; it does not replace an operational account of scheduling.

Three directions seem especially natural from this point.  One is to develop abstract domains for compact residual meanings beyond finite Boolean quotients, for example domains recording resource bounds, recurrent control components, queue sizes, pending obligations, or name-generation profiles.  A second is to look for useful sufficient conditions under which process operators descend to compact meanings; the descent criterion gives the exact abstract condition, while the CCS examples give both positive prefix-choice laws and a negative contextuality result.  A third is to extend the residual-process application to calculi with binding and mobility, where a nominal process space, name-invariant observations, freshness, restriction and scope extrusion would all have to be handled explicitly.

There is also a more algebraic line of development.  When the observation space carries semigroup or monoid operations, one can ask whether these operations extend to the compactified residual space in a useful way.  The existing algebra of \(\beta S\) for discrete semigroups suggests possible classifications of recurrent or escaping program behaviour, although such a classification would require hypotheses not needed for the collecting semantics developed here.

\section{Conclusion}
\label{sec:conclusion}

A filtered run
\[
  a:T\to X
\]
into a Tychonoff observation space has the compact tail meaning
\[
  \M_X^{\Filt}(a)
  =
  \beta a[\Ext(\Filt)]
  =
  \bigcap_{F\in\Filt}\cl_{\beta X}a[F].
\]
For sequential time this is the tail-cluster set of the stream in the Stone--\v{C}ech compactification of the observation space.  The same formula covers ordinary convergence, finite recurrence, mixed recurrence and escape, and unbounded behaviour in noncompact observation spaces.

The useful features of the construction are simple.  Meanings are invariant under changes on small sets of times.  Continuous observations commute with meanings, so finite quotients, modal observation quotients and resource observations are images of the same compact tail object.  Clopen regions of the compactification have a direct temporal reading: containment gives eventual truth, and nonempty intersection gives recurrence.  Strengthening the time filter gives fairness-refined meanings, and compactifying products preserves asymptotic correlations between simultaneous observations.

For CCS, infinite executions become streams of residual processes.  The resulting semantics separates stable looping, finite recurrent divergence, mixed recurrence and unbounded residual growth.  It validates residual-tail laws for finite prefixing, guarded unfolding, finite choice and finite prefix-choice forms.  The parallel-composition example supplies the boundary: these are tail laws, not full contextual equivalences for CCS, because synchronisation may use a finite prefix before it vanishes from the residual tail.

The compact space need not be enumerated.  Its practical content is obtained by applying observations.  Finite Boolean quotients turn compact residual meanings into recurrent abstract states, and finite transition graphs reduce these meanings to strongly connected component calculations.  Under the edge-fair reading of \Cref{prop:edge-fair-finite-graph}, the finite meanings are reachable terminal components.  Resource observables detect unbounded escape by mapping compact residual meanings into remainders such as \(\N^*\).

Stone--\v{C}ech compactification therefore gives a compact collecting semantics for residual process tails, while finite and resource-level observations provide the operational interface.  Once the observation topology is fixed, the same construction handles tail limits, recurrence, fairness, finite abstraction, relational correlation and unbounded escape.

\section*{Acknowledgement of AI assistance}
The author used an AI assistant during the preparation of this manuscript, including for copy-editing, formatting, and discussion of exposition.  The author is responsible for the final content, including all mathematical statements and proofs.

\end{document}